\documentclass[journal]{IEEEtran}
\usepackage{graphicx}  
\usepackage{caption}
\usepackage{epstopdf}  
\usepackage{cite}
\usepackage{amssymb}
\usepackage{amsmath}
\usepackage{amsfonts}
\usepackage{algorithmic}
\usepackage{algorithm}
\usepackage{amsthm}
\usepackage{booktabs}
\usepackage{lipsum}
\usepackage{cuted}
\usepackage{textcomp}
\usepackage{verbatim}
\usepackage{array}  
\usepackage{float}    
\usepackage{enumitem}
\usepackage{bm}
\usepackage{multicol}
\usepackage{subfig}
\usepackage{hyperref}
\usepackage{xcolor}

\definecolor{steelblue}{RGB}{70, 130, 180}
\newtheoremstyle{mytheoremstyle}
  {0pt}
  {0pt}
  {\normalfont}
  {1em}
  {\itshape}
  {.}
  {.5em}
  {}

\theoremstyle{mytheoremstyle}

\newtheorem{myTheo}{Theorem}
\newtheorem{myAssu}{Assumption}

\begin{document}
	\setlength{\abovedisplayskip}{2pt}
	\setlength{\belowdisplayskip}{2pt}
	
	\title{AIGC-assisted Federated Learning for Vehicular Edge Intelligence: Vehicle Selection, Resource Allocation and Model Augmentation}
	
	\author{Xianke Qiang, Zheng Chang,~\IEEEmembership{Senior~Member,~IEEE,} Geyong Min,~\IEEEmembership{Senior~Member,~IEEE}
		
		\thanks{X. Qiang and Z. Chang are with School of Computer Science and Engineering, University of Electronic Science and Technology of China, Chengdu 611731, China. G. Min is with Department of Computer Science, University of Exeter, Exeter, EX4 4QF, U.K. 
	}}
	
	\maketitle
	
	\begin{abstract}
		To leverage the vast amounts of onboard data while ensuring privacy and security, federated learning (FL) is emerging as a promising technology for supporting a wide range of vehicular applications. Although FL has great potential to improve the architecture of intelligent vehicular networks, challenges arise due to vehicle mobility, wireless channel instability, and data heterogeneity. To mitigate the issue of heterogeneous data across vehicles, artificial intelligence-generated content (AIGC) can be employed as an innovative data synthesis technique to enhance FL model performance. In this paper, we propose AIGC-assisted Federated Learning for Vehicular Edge Intelligence (GenFV). We further propose a weighted policy using the Earth Mover's Distance (EMD) to quantify data distribution heterogeneity and introduce a convergence analysis for GenFV. Subsequently, we analyze system delay and formulate a mixed-integer nonlinear programming (MINLP) problem to minimize system delay. To solve this MINLP NP-hard problem, we propose a two-scale algorithm. At large communication scale, we implement label sharing and vehicle selection based on velocity and data heterogeneity. At the small computation scale, we optimally allocate bandwidth, transmission power and amount of generated data. Extensive experiments show that GenFV significantly improves the performance and robustness of FL in dynamic, resource-constrained environments, outperforming other schemes and confirming the effectiveness of our approach.
	\end{abstract}
	
	\begin{IEEEkeywords}
		Artificial intelligence generated content, non-IID, data augmentation, resource allocation, vehicle selection.
	\end{IEEEkeywords}

\section{Introduction}
    In the era of ubiquitous intelligence, deep model training is vital for extracting insights from vehicular data to improve AI-driven services \cite{10056271}. The vast amount of data collected by vehicles provide a strong foundation for developing intelligent in-vehicle services. By leveraging the onboard data, Machine Learning (ML) has demonstrated its potential across various Intelligent Transportation System (ITS) applications \cite{8345672}, including traffic sign classification, object detection, congestion prediction, and more. However, the sensing capabilities and onboard computational power of vehicles remain limited. Additionally, offloading raw data to edge servers introduces significant privacy concerns and demands substantial bandwidth. As a result, there is an urgent need for privacy-preserving, distributed ML solutions in modern vehicular networks to enable higher levels of automation on the road, where vehicles must make swift operational decisions \cite{10121038}.\par 
    Federated Learning (FL) \cite{mcmahan2017communication} offers a promising solution by allowing vehicles to collaboratively train a global model while preserving data privacy. In FL, vehicles transmit only local model updates, rather than raw data, to a Roadside Unit (RSU) or Base Station (BS) \cite{10144680}. This decentralized approach significantly reduces the risk of data breaches, alleviates bandwidth constraints, and enhances the scalability and efficiency of vehicular networks.\par

        \begin{figure}[t]
            \centering
            \includegraphics[width=0.40\textwidth]{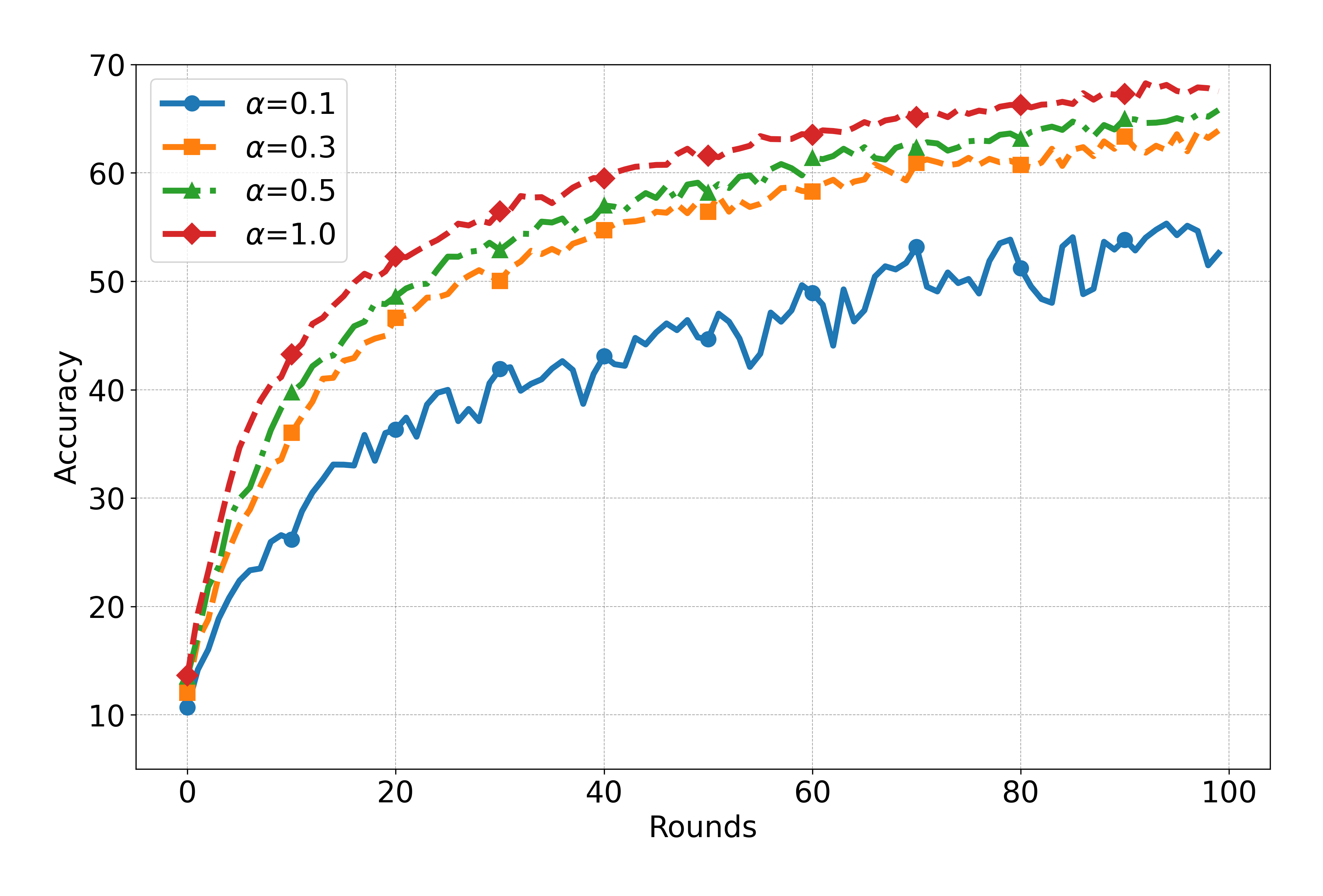}
            \caption{The impacts of data distribution $Dir(\alpha)$ on the testing accuracy.}
            \label{fig:nonIID}
        \end{figure}
    Despite its privacy advantages, FL faces significant challenges, particularly in handling non-independently and identically distributed (non-IID) data distributions \cite{li2022federated, zhao2018federated}. As illustrated in Fig. \ref{fig:nonIID}, a well-balanced Dir(1.0) distribution achieves faster convergence and higher accuracy compared to Dir(0.1), underscoring the critical need to address non-IID data issues. These challenges become even more pronounced in vehicular networks, where applications such as image classification and object detection \cite{ma2020artificial, 10313052} must contend with the dynamic nature of participating vehicles and the diversity in their data. The frequent turnover of vehicles and the inherent non-IID characteristics of their data further complicate the training process, making model convergence and performance optimization more challenging.\par
	Effectively addressing non-IID data distribution is crucial for deploying FL in vehicular edge intelligence (VEI). AI-generated content (AIGC) services, such as Stable Diffusion (SD) \cite{rombach2022high}, DALL-E2 \cite{ramesh2022hierarchical}, and Imagen \cite{saharia2022photorealistic}, offer promising solutions by augmenting local datasets. Recent studies have demonstrated the potential of generative AI in mitigating the effects of non-IID data by enriching training datasets \cite{10521828, 10557146, Morafah2024StableDD, ahn2022communication, li2024filling}. However, vehicles are primarily designed for transportation, and their capacity or willingness to generate additional data for FL tasks, particularly in high-mobility scenarios, remains limited. Therefore, it is essential to design a new architecture for AIGC-assisted FL system in a way that minimizes the burden on vehicles, alleviates non-IID issues, and safeguards privacy. However, realizing such a framework involves overcoming several system-level, resource allocation, and data management challenges:
    \begin{itemize}
        \item System-level perspective: \textit{How can we design an AIGC-assisted FL architecture that accommodates vehicle mobility and time-sensitive tasks?} As image generation is resource-intensive and not all vehicles can generate data, the framework must also address connectivity disruptions caused by high mobility to ensure adaptability to time-sensitive tasks in dynamic vehicular environments.
        \item Resource allocation perspective: \textit{How can we efficiently allocate communication and computation resources in AIGC-assisted FL for VEI?} This requires optimizing bandwidth and computational resource allocation while accounting for network latency, vehicle mobility, resource constraints, and communication instability to ensure efficient system performance.
        \item Data management perspective: \textit{How can we determine which labels should be generated in each round, and how much synthetic data is needed? }  An effective strategy must adapt to data heterogeneity, prevent model bias, account for resource constraints, and ensure generated data complements real data to meet global model needs.
    \end{itemize}
	In this paper, we introduce a novel framework called \textbf{AIGC-assisted Federated Learning for Vehicular Edge Intelligence (GenFV)}, which leverages the AIGC to address the non-IID data challenge commonly faced in FL. To the best of our knowledge, this is the first research on AIGC-assisted FL for achieving VEI, with a focus on addressing the non-IID challenges. Through GenFV, we aim to significantly improve the performance and robustness of FL in mobile, resource-constrained environments. The main contributions of this paper are as follows:
	
	\begin{itemize}
		\item  We propose the GenFV framework, which integrates AIGC into FL. By employing a mobility- and data-aware vehicle selection strategy followed by resource allocation, GenFV effectively mitigates data heterogeneity and enhances global model performance. This approach mitigates the impact of vehicle mobility and data heterogeneity on FL model training.
		\item We formulate a delay minimization multi-objective problem that incorporates vehicle selection, bandwidth allocation, transmission power allocation, and data generation. The problem is modeled as a mixed-integer non-linear programming (MINLP) problem, which is NP-hard and challenging to solve directly due to its non-convex nature.
		\item We design a two-scale algorithm to address the formulated delay minimization problem. At the large communication scale, we perform label sharing and vehicle selection. At the smaller computation scale, we decompose and solve three subproblems: bandwidth allocation, transmission power assignment, and model augmentation through data generation.
		\item We evaluate the effectiveness of the proposed GenFV framework through extensive simulations using various open datasets. The results demonstrate the framework's ability to enhance model performance while addressing the challenges of non-IID data, mobility, and resource constraints.
	\end{itemize}
	    The rest of the article is organized as follows. In Section.\ref{section2}, the related works are introduced. Section \ref{section3} presents the architecture of GenFV, followed by weighted policy and convergence analysis. Section \ref{section4} conducts latency-energy analysis and then formulate time minimization problem. Section \ref{section5} provides the problem transformation and solution. Section \ref{section6} evaluates the performance of our proposed solution through extensive simulations. Finally, Section \ref{section7} gives an conclusion.

\section{Related Work}
\label{section2}
	The non-IID challenge of FL was first introduced by \cite{mcmahan2017communication} and subsequently impacted the convergence and performance of the global model \cite{zhao2018federated}. To address this issue, numerous methods have been proposed. One such method, FedProx \cite{li2020federated}, incorporates a proximal term to constrain local updates relative to the global model.  However, these methods do not fundamentally resolve the non-IID problem and face performance bottlenecks in extreme cases with highly imbalanced data distributions.\par
	Data augmentation methods provide an effective approach to address the issue of non-IID data distribution in FL. By utilizing AIGC service, clients in FL can conduct data synthesis to mitigate data heterogeneity issue \cite{li2024filling}. SlaugFL\cite{10521828} is a selective GAN-based data augmentation scheme for federated learning that enhances communication efficiency and model performance by using representative devices to share local class prototypes. IMFL-AIGC\cite{10557146} addresses the challenge of client participation in AIGC-empowered FL by proposing a data quality assessment method and an incentive mechanism that minimizes server costs while enhancing training accuracy. Gen-FedSD\cite{Morafah2024StableDD} utilizes text-to-image models to generate high-quality synthetic data for clients in federated learning, effectively addressing non-IID data distribution challenges. FedDif\cite{ahn2022communication}, a novel diffusion strategy for federated learning that addresses the weight divergence challenge caused by non-IID data by allowing users to share local models via device-to-device communications. In AIGC-assisted federated learning, optimizing the allocation of communication and computation resources is a key challenge, particularly in high-mobility environments. FIMI\cite{li2024filling} formulates an optimization problem aimed at minimizing overall energy consumption on the device side, considering four variables: the number of additional synthesized data for client, computation frequency, sub-bandwidth, and transmission power. However, many studies assume that end devices possess sufficient computing power to generate images. In reality, not all devices have this capability, nor do they necessarily need to create synthetic data. Additionally, given that devices are constantly in motion, it is crucial to consider how to effectively leverage AIGC to enhance data in dynamic scenarios.
    
\section{GenFV: AIGC-assisted Federated Learning for Vehicular Edge Intelligence}
\label{section3}
	In this section, we first introduce GenFV architecture and workflow. Then we introduce the AIGC. Finally, we propose FL weighted policy and convergence analysis.
    \subsection{Architecture and Workflow}
	The proposed GenFV is illustrated in Fig. \ref{fig:Architecture}. We consider a general VEC system that includes one server which is deployed on RSU and a set of vehicles $\{1,2,\dots, V\}$. The set of available vehicles within the communication range of the FL server at round $t$ is denoted by $\mathcal{N}$ which satisfies $\mathcal{N} \subset \mathcal{V}$. The dataset of the vehicle $n$ is denoted as $\mathcal{D}_n = \{\mathcal{X}_n, \mathcal{Y}_n\}$, where $\mathcal{X}_n = \{x_{n}^{1}, x_{n}^2, \dots, x_{n}^{\left|\mathcal{D}_n\right|}\}$ is the training data, $\mathcal{Y}_n = \{y_{n}^1, y_{n}^2,\dots,y_{n}^{\left|\mathcal{D}_n\right|}\}$ represents the corresponding labels, and $\left|\mathcal{D}_n\right|$ is the number of training data samples of vehicle $n$. The GenFV process involves five main steps as follows:\par
    \subsubsection{Labels Sharing}
    First, all vehicles within the communication range of the RSU share low-privacy metadata with the RSU, such as transmission power, computational frequency, data quality metric, and non-sensitive data labels. Then RSU can either perform text-to-image generation or directly generate images based on a pre-trained or fine-tuned model. Once the RSU receives the labels and associated metadata from the vehicles, it applies a two-scale algorithm for vehicle selection and resource allocation in Section \ref{section5}. Additionally, the RSU leverages the shared labels to evaluate the data heterogeneity across vehicles, enabling more accurate guidance for data generation. It is worth noting that the local roadside infrastructure, such as the RSU of the traffic department, can access to the labels for the task \cite{10229176}. Therefore, label sharing primarily helps the RSU better assess and address data heterogeneity among vehicles without compromising data privacy.  
    \subsubsection{Model Distribution}
    At the start of each training round, the global model is distributed to vehicles as the initial local model. 
    \subsubsection{Local Model Training} 
    After the selected vehicles receive the distributed global model, they start training their local models using their local datasets. 
    \subsubsection{Local Model Upload}  
    After completing local model training, vehicles upload their models to the RSU for aggregation. To enhance transmission efficiency, vehicles employ optimized power and bandwidth allocation strategies, ensuring rapid model uploads with limited communication resources.  
    \subsubsection{Model Augmentation and Aggregation}  
    While vehicles perform local model training, the RSU generates images based on data generation strategy and uses these generated data to train an augmented model. This approach allows the RSU to effectively utilize idle computational resources to improve its model performance during the waiting period for vehicle model uploads. Once the RSU receives the locally updated models from the selected vehicles, the training of the augmented model in RSU is also down. Subsequently, the RSU aggregates the local models with the enhancement model based on the weighted strategy proposed in Section \ref{section3c}, significantly improving the global model's performance and robustness. 
    
    \begin{figure}[t]
        \centering
        \includegraphics[width=0.90\linewidth]{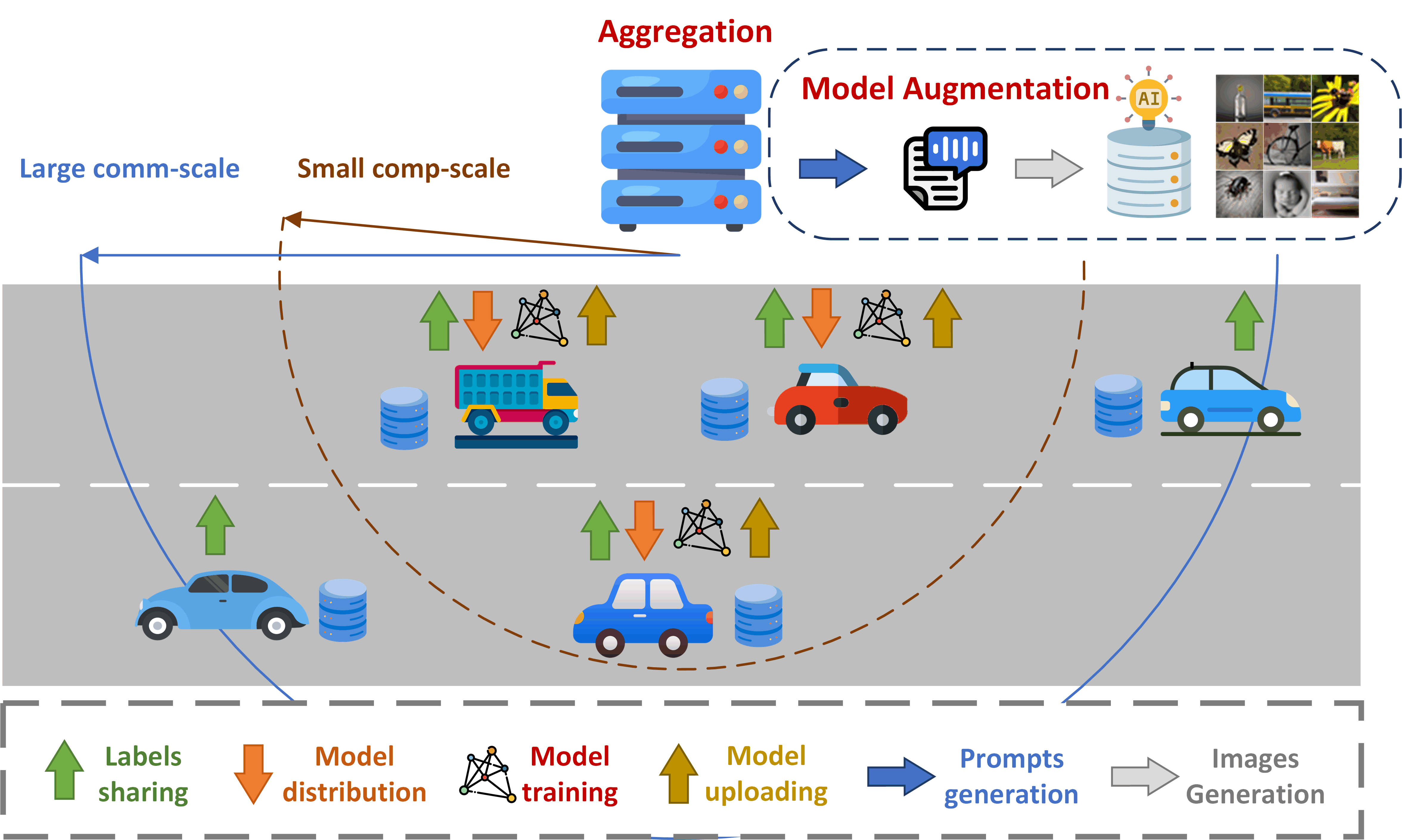}
        \caption{System model.}
        \label{fig:Architecture}
    \end{figure}
    
    \subsection{AIGC Model}
    Diffusion models (DMs) have gained increasing interests in the field of generative modeling. DMs gradually transform a pristine image into pure noise through forward diffusion \cite{ho2020denoising,nichol2021improved}. Subsequently, they learn to reconstruct the original clean image from this pure noise through reverse diffusion. 
    
    In the forward diffusion process, a Gaussian noise with zero mean and variance $\lambda_t$ is slowly added in $T$ time-steps to degrade a clean image $(x_0 \sim q(x_0))$ and ultimately converge into pure Gaussian noise at step $T(x_T)$. The noisy image at each time-step $t$ from the previous time-step $t-1$ can be derived using the following equation:
    \begin{equation}
    q\left(x_t \mid x_{t-1}\right)=\boldsymbol{N}\left(x_t;\sqrt{1-\lambda_t} x_{t-1},\lambda_t I\right),
    \end{equation}
    where $t \in [1,T]$ and $0 \leq \lambda_{1:T} \leq 1$ denote the noise scale scheduling. In the reverse diffusion process, a noise predictor model $\varepsilon_\theta$ is trained to predict the noise $\varepsilon_t$ at each time-step $t$ using the following loss function:
    \begin{equation}
    \mathcal{L}=\mathbb{E}_{t, x_0, \varepsilon}\left[\left\|\varepsilon_t-\varepsilon_\theta\left(x_t, t\right)\right\|_2^2\right].
    \end{equation}

    \subsection{FL Weighted Policy and Convergence Analysis}
        \subsubsection{Weighted Policy}
    	\label{section3c}
    	In this paper, we denote the loss function for each data sample $i$ as $\ell(\boldsymbol{\omega},x_n^{i})$. For each dataset $\mathcal{D}_n$ corresponding to vehicle $n$, the local loss function is defined as $L_n(\boldsymbol{\omega}) = \frac{1}{\left|\mathcal{D}_n\right|} \sum_{i=1}^{\left|\mathcal{D}_n\right|} \ell(\boldsymbol{\omega},x_{n}^i)$. Similarly, the loss function of the augmented model trained by generated images is given by $L_a(\boldsymbol{\omega}) = \frac{1}{\left|\mathcal{D}_a\right|} \sum_{i=1}^{\left|\mathcal{D}_a\right|} \ell(\boldsymbol{\omega},x_a^i)$.\par
        The local model $\boldsymbol{\omega}_n^{t}$ of vehicle $n$ for round $t$ at $t_l-1$ steps update with following rule:
            \begin{equation}
                \boldsymbol{\omega}_n^{t,t_l} =
                \begin{cases}
                    \boldsymbol{\omega}_n^{t,t_l-1} - \eta \nabla L_n(\boldsymbol{\omega}_n^{t,t_l-1}), & \text{if }t_l\mod h \neq 0, \\
                    \boldsymbol{\omega}^{t}, & \text{if }t_l\mod h = 0,
                \end{cases}
                \nonumber
            \end{equation}
        and the AIGC-assistant augmented model $\boldsymbol{\omega}_a^{t}$ trained with AIGC-generated images has the similar updating rule:
            \begin{equation}
                \boldsymbol{\omega}_a^{t,t_l} =
                \begin{cases}
                    \boldsymbol{\omega}_a^{t,t_l-1} - \eta \nabla L_a(\boldsymbol{\omega}_a^{t,t_l-1}), & \text{if }t_l\mod h \neq 0, \\
                    \boldsymbol{\omega}^{t}, & \text{if }t_l\mod h = 0.
                \end{cases}
                \nonumber
            \end{equation}
        where $\boldsymbol{\omega}^{t} = \kappa_1 \sum_{\forall n \in \mathcal{N}^t} \rho_n (\boldsymbol{\omega}_n^{t,h-1}-\eta \nabla L_n(\boldsymbol{\omega}_n^{t,h-1})) + \kappa_2 (\boldsymbol{\omega}_a^{t,h-1}-\eta \nabla L_n(\boldsymbol{\omega}_a^{t,h-1}))$. Here,  $\rho_n = \frac{\left|\mathcal{D}_n\right|}{\sum_{\forall n \in \mathcal{N}^t} \left|\mathcal{D}_n\right|}$, $\kappa_1$ represents the proportion of the FL model, $\kappa_2$ corresponds to the proportion of the AIGC-assistant augmented model, and $\boldsymbol{\omega}_n^{t,0} = \boldsymbol{\omega}_a^{t,0} = \boldsymbol{\omega}^{t-1}$.\par
        To characterize the data quality of vehicle $n \in \mathcal{N}$, we use the criterion $\|\nabla L_n(\boldsymbol{\omega})-\nabla L(\boldsymbol{\omega})\| \leq \lambda_n$, where $\lambda_n$ is noted as the upper bound of the gradient difference between the local loss function and the global loss function. A larger $\lambda_n$ indicates poorer data quality\cite{10557146,ahn2022communication}. Specifically, $\lambda_n$ can be estimated by 
    	\begin{align}
    		&\|\nabla L_n(\boldsymbol{\omega}) - \nabla L(\boldsymbol{\omega})\| \nonumber\\
    		\leq&\sum_{i=1}^{Y} \|[p^n(y=i)-p(y=i)]\mathbb{E}_{x|y=i}l_i(x,\boldsymbol{\omega}) \nonumber\\
    		\leq&EMD_n \cdot g_n =\lambda_n,
              \label{upperbound}
    	\end{align}
    	where $EMD_n = \sum_{i=1}^{Y}\|p_n(y=i)-p(y=i)\|$ and $g_n=max_{i\in Y}\|\mathbb{E}_{x|y=i}l_i(x,\boldsymbol{\omega})\|$ is widely-adopted average earth mover's distance (EMD) as the data distribution heterogeneity among vehicles. For vehicle $n$ with non-IID data distribution, the data synthesis process is performed as follows: for each label $i$, given the proportion of data samples with label $i$ in the global data (denoted as $p(y=i)=\frac{1}{Y}$), where $Y$ is the number of labels). The proportion of data samples with label $i$ in $\mathcal{D}_n$ (denoted as $p_n(y=i)$).\par
        For all participate vehicles, the average data quality is denoted as $\bar{EMD} = \frac{\sum_{n=1}^{\mathcal{N}}EMD_n}{|\mathcal{N}|}$. As shown in (\ref{upperbound}), a larger $EMD_n$ results in a larger upper bound, which in turn leads to poorer local model performance\cite{zhao2018federated}. Therefore, we aim to assign a smaller proportion to Fed model with a larger $\bar{EMD}$ in the weighted policy. Thus, we denote the $\kappa_1 = 1- (\frac{\bar{EMD}}{2})^2$ and $\kappa_2=(\frac{\bar{EMD}}{2})^2$ as the model quality of the Fed model and augmented model respectively.Thus, we denote the $1- (\frac{\bar{EMD}}{2})^2$ as the model quality of the Fed model. The global model update strategy is denoted as
    	{\begin{align}        
    		\boldsymbol{\omega}^{t} =(1- (\frac{\bar{EMD}}{2})^2)\sum_{\forall n \in \mathcal{N}^t}\rho_n \boldsymbol{\omega}_n^{t} + (\frac{\bar{EMD}}{2})^2 \boldsymbol{\omega}_a^{t},
    		\label{modelaugmentation}
    	\end{align}}
        where $\boldsymbol{\omega}_n^{t} = \boldsymbol{\omega}_n^{t,h}$ and $\boldsymbol{\omega}_a^{t} = \boldsymbol{\omega}_a^{t,h}$.

    \subsubsection{Convergence Analysis}
    	To analyze the convergence rate, we first make the assumptions as follows \cite{10557146,9069945}:\par
            \begin{myAssu}
                $L_1, \ldots, L_n$ are all $\beta-Lipschiz$, i.e., for all $\boldsymbol{\omega}$ and $\boldsymbol{v}$, $\|L_n(\boldsymbol{\omega}) - L_n(\boldsymbol{v})\| \leq \beta \|\boldsymbol{\omega}-\boldsymbol{v}\|$.
                \label{assump-1}
            \end{myAssu}
            \begin{myAssu}
                $L_1, \ldots, L_n$ are all $\varrho-Lipschiz$ smooth, i.e., for all $\boldsymbol{\omega}$ and $\boldsymbol{v}$, $\|\nabla L_n(\boldsymbol{\omega}) - \nabla L_n(\boldsymbol{v})\| \leq \varrho \|\boldsymbol{\omega}-\boldsymbol{v}\|$.   
                \label{assump-2}
            \end{myAssu}
            \begin{myAssu}
                $L_n(\boldsymbol{\omega})$ satisfies $\mu-strong$ convex, i.e., for all $\boldsymbol{\omega}$, $\|L_n(\boldsymbol{\omega}) - L_n(\boldsymbol{\omega^*})\| \leq \frac{1}{2\mu} \|L(\boldsymbol{\omega})\|$, where $\boldsymbol{\omega^*}$ is the optimal parameter.
                \label{assump-3}
            \end{myAssu}
            \begin{myAssu}
                The bound of data quality: $\|\nabla L_n(\boldsymbol{\omega})-\nabla L(\boldsymbol{\omega})\| \leq \lambda_n$;  $\|\nabla L_a(\boldsymbol{\omega})-\nabla L(\boldsymbol{\omega})\| \leq \lambda_a$.
                \label{assump-4}
            \end{myAssu}
            \begin{myAssu}
                Let $x_t^n$ present the random sample dataset from the vehicle $n$.  The variance of stochastic gradients in each vehicle is bounded: $\mathbb{E}\left\|\nabla L_n\left(\boldsymbol{\omega}_n, \xi_n\right)-\nabla L_n\left(\boldsymbol{\omega}\right)\right\|^2 \leq$ $\delta_n^2$, for $n=1, \ldots, N$.
                \label{assump-5}
            \end{myAssu}
            \begin{myTheo} Let Assumptions 1 to 5 hold, we assume the global iterations $T$, the set of participate vehicles is denoted as $\mathcal{N} \in \mathcal{V}$. By setting $\eta < \frac{1}{\varrho}$, the convergence upper bound is given as:
            \begin{equation}
                L(\boldsymbol{\omega}(T,Th)) - L(\boldsymbol{\omega}^{*})\leq \chi^{hT}\Theta+(1-\chi^{hT})\psi\Lambda,
            \end{equation}
            where $\Theta = L(\boldsymbol{\omega}(0,0)) - L(\boldsymbol{\omega}^*)$, $\Lambda=\kappa_1 \sum_{n \in \mathcal{N}} \rho_n(\sigma_n+\lambda_n) + \kappa_2 \lambda_a$, $\chi = 1-2\mu\eta+2\mu\varrho\eta^2$, $\psi=\frac{\beta (\eta\varrho+1)^h -1}{\varrho(1+\chi^h)}$.
            \end{myTheo}
            \begin{proof}
                The detailed proof is provided in the Appendix.
            \end{proof}

    \begin{figure}[t]
        \centering
        \includegraphics[width=0.90\linewidth]{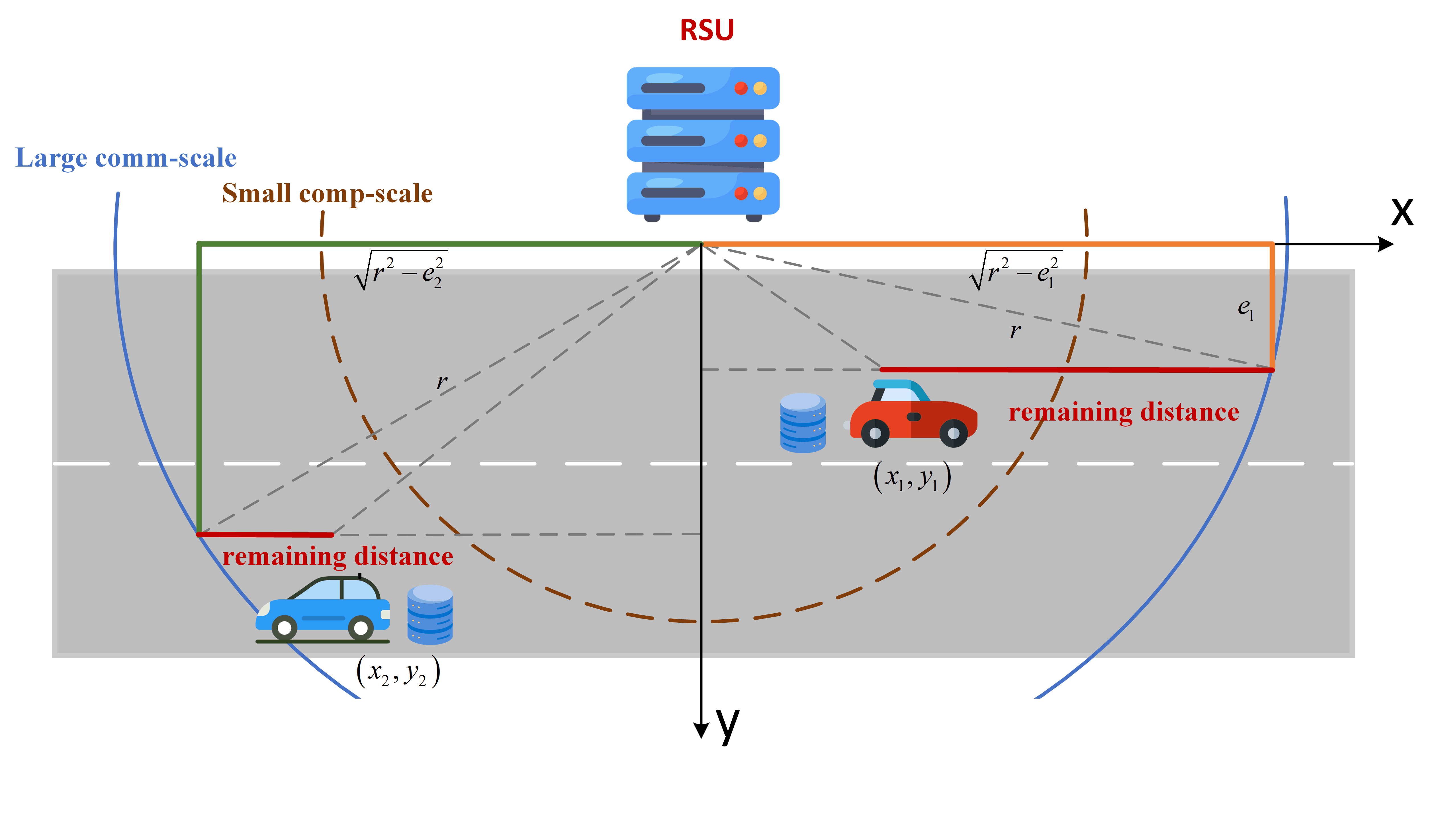} 
        \caption{Vehicle Model}
        \label{fig:vehicle} 
    \end{figure}    

\section{Latency-Energy Analysis and Problem Formulation}
\label{section4}
	In this section, we first analyze the latency and energy of proposed GenFV. Next, we present the decision variables. Finally, we formulate the system time minimization problem.
	
	\subsection{Latency-Energy Analysis}
	\subsubsection{Labels sharing}
	First, all vehicles within the communication range of the RSU share low-privacy metadata with the RSU, such as transmission power, computational frequency, and non-sensitive data labels. Since the amount of data transmitted during this phase is significantly smaller than the model parameters, the associated delay and energy consumption can be considered negligible.
	
	\subsubsection{Model Distribution}
	After the RSU receives the information tuple, the RSU selects vehicles $\mathcal{N}^t$ according to the information and vehicle selection strategy in $t$-th round. Then the RSU broadcasts the model to selected vehicles. We consider the broadcast for downlink transmission during this process. The downlink delay is considered insignificant in comparison to the uplink delay, as presented in, due to the RSU's significantly higher downlink power and the allocation of a larger downlink bandwidth for data distribution \cite{yang2020energy}. 
	
	\subsubsection{Local Model Training}
	GPUs are the most commonly used accelerators for DNN computations \cite{chen2023service},\cite{chen2022energy},\cite{8057205}. We consider the GPUs instead of CPUs in this work for two reasons. First, CPUs cannot support relatively large and complicated model training tasks. Second, GPUs are more efficient than CPUs for on-device DNN training and are increasingly integrated into today’s mobile devices. GPU execution time model is formulated as 
	\begin{equation}
		T_n^{cp} = t_n^0 + \frac{c_1 b_n \theta_n^{mem}}{f_n^{mem}} + \frac{c_2 b_n \theta_n^{core}}{f_n^{core}},
	\end{equation}
	where $t_n^0$ represents the other component unrelated to training task; $\theta_n^{mem}$ and $\theta_n^{core}$ denote the number of cycles to access data from the memory and to compute one mini-batch size of data samples, respectively. For simplicity, we assume that the number of cycles for data fetching and computing scales, $c_1$ and $c_2$. The $b_n$ represents the number of batches. GPU runtime power model is modeled as a function of the core/memory voltage/frequency:
	\begin{equation}
		p_n^{cp} = p_n^{G0} + \zeta_n^{mem} f_n^{mem} + \zeta_n^{core} (V_n^{core})^2 f_n^{core},
	\end{equation}
	where $p_n^{G0}$ is the summation of the power consumption unrelated to the GPU voltage/frequency scaling; $V_n^{core},f_n^{core},f_n^{mem}$ denote the GPU core voltage, GPU core frequency, and GPU memory frequency, respectively; $\zeta_n^{mem}$ and $\zeta_n^{core}$ are the constant coefficients that depend on the hardware and arithmetic for one training iteration, respectively. With the above GPU power and performance model, the local energy consumed to pass a single mini-batch SGD is the product of the runtime power and the execution time:
	\begin{equation}
		E_n^{cp} = p_n^{cp} \cdot T_n^{cp}.
	\end{equation}
	
	\subsubsection{Model uploading}
	We consider using orthogonal frequency-division multiple access technique as data access method for local gradient uploading. Since the communication resources are limited, we assume each RSU allocates $M$ subcarriers for selected $N$ vehicles ($M<N$) to uploading, and each vehicle can have at most one subcarrier. Therefore, the uplink transmission rate of vehicle $n$ is. 
	\begin{equation}  
		r_n^U = l_n W log_2 (1+\frac{\phi_n h_0 d_n^{-\gamma}}{N_0}),
		\label{eq:r_n^U}
	\end{equation}
	where $l_n = \sum_{m=1}^{M}\beta_{n,m}$ is the number of allocated subcarriers. The subcarrier $m$ is assigned to the vehicle $n$ when $\beta_{n,m}=1$ and otherwise when $\beta_{n,m}=0$. Note that $\beta_{n,m}\in\{0,1\}$ and $\sum_{m=1}^{M}\beta_{n,m}=1$, so $l_n \in\{0,1\}$. $W$ denotes the bandwidth of the uplink channel, $\phi_n$ is the transmission power of vehicle $n$, $h_0$ is the channel gain at a unit distance, $d_n^{-\gamma}$ is the distance between vehicle $n$ and the base station and $\gamma$ is the path loss exponent, $N_0$ is the noise power. Once vehicles finished the model training, vehicles upload their training model to RSU for aggregation. The time delay and energy consumption is defined as:
	\begin{align}
		T_n^{mu} &= \frac{s(\boldsymbol{\omega})}{r_n^U},\\
		E_n^{mu} &= \phi_n \frac{s(\boldsymbol{\omega})}{r_n^U},
	\end{align}
	
	\subsubsection{Model augmentation and Aggregation}
		While the vehicles train their models locally, the RSU generates images. We need to consider both the image generation phase and the augmented model training phase. For image generation phase, we define $d_{m,t}$ to represent the computing capability required for each inference step $t$. The latency of inference executed at RSU can be expressed as follows:
    	\begin{equation}
    		T^{inf}_s = b \sum_{t = 1}^{I} \frac{d_{m,t}}{f_{rsu}} = b^t t_0,
    	\end{equation}
    	where $I$ is the overall inference steps for generating one image and $f_{rsu}$ is the inference capacity of RSU and $t_0 = \sum_{t = 1}^{I}\frac{d_{m,t}}{f_{rsu}}$ means the time delay of generating one image.\par
    	The generated images dataset $\mathcal{D}_s$ ($b_s$ batches) is utilized for training the augmented model. The training delay can be computed as follows:
    	\begin{equation}
    		T_s^{cp} = t_s^0 + \frac{c_s^1 b_s ^t\theta_s^{mem}}{f_s^{mem}} + \frac{c_s^2 b_s^t \theta_s^{core}}{f_s^{core}},
    	\end{equation}	
        Upon receiving the local weight parameter vectors from all devices, along with the augmented model trained by generated datasets in RSU, the RSU aggregates them to compute the new global weights using the weighted strategy proposed in Section \ref{section3c}. Since the RSU typically possesses robust computational capabilities and a stable power supply, the latency of the aggregation phase is relatively low compared to other stages of the process.\par 

    \subsubsection{Overall Latency}
    The overall training latency of vehicle $n$ comprises the local model computation delay and the local model uploading delay, defined as:
	\begin{equation}
		T_n = T_n^{cp} + T_n^{mu}.
		\label{overallT}
	\end{equation}
	
	\subsection{Decision Variables in GenFV}
	In the GenFV, the following decision variables should be determined.
	\subsubsection{Vehicle selection} At each training round, the vehicle selection decision, denoted by $\alpha$, is determined based on the online vehicles' channel conditions, computing capabilities and local datasets heterogeneity, denoted by binary vector $\boldsymbol{\alpha}^t$. Each element $\alpha_n^t$ represent the vehicle selection index at round $t$ for vehicle $n$, constrained by
	\begin{equation}
		\alpha_n^t \in \{0,1\}, \forall n \in \mathcal{N},
		\label{alpha_cons}		
	\end{equation}
	where $\alpha_n^t=1$ indicates that device $n$ is selected; and $\alpha_n^t=0$ otherwise. In each round, the selected vehicles set is denoted as
	\begin{equation}
		\mathcal{N}^t = \boldsymbol{\alpha^t} \mathcal{N}.
	\end{equation}
	\subsubsection{Bandwidth allocation}
	In each training round, we consider the OFDMA for data transmission. Let $\boldsymbol{l}^t$ denote the subcarrier allocation decision. Then each element $l_n^t$ constrined by 
	\begin{equation}
		l_n^t \in \{0,1\}, \forall n \in \mathcal{N}^t,
		\label{bandwidth_cons1}
	\end{equation} 
	where $l_n^t=1$ indicates that device $n$ is selected; and $l_n^t=0$ otherwise. Note that the number of allocated subcarriers should not exceed the overall spectrum capacity, i.e.,
	\begin{equation}
		\sum_{n=1}^{\mathcal{N}^t} l_n^t \leq M.
		\label{bandwidth_cons2}
	\end{equation}
	\subsubsection{Transmission power}
	In each round, the uplink transmission power decision, denoted by $\boldsymbol{\phi}^t$, is decided by channel conditions and participating vehicles. The decision is constrained by 
	\begin{equation}
		\phi_{min} < \phi_n^t <\phi_{max}, \forall n \in \mathcal{N}^t,
		\label{phi_cons}
	\end{equation}
	where $\phi_{min}$ is the minimal value of uplink transmission power, $\phi_{max}$ is the maximal value of $\phi_n$.
	
	\subsubsection{Data generation}
	At each round, the generated image amount decision is denoted as $\boldsymbol{b}^t$. The number of the amount is a natural number, so the constraints is as following:
	\begin{align}
		b^t \in \mathbb{N}.
		\label{generation_cons1}
	\end{align} 
	Besides, RSU generates images while vehicles training local models, and the generating time will not exceed the FL training time. 
	\begin{equation}
		T_s^{inf} + T_s^{cp} \leq \bar{T},\forall n \in \mathcal{N}^t.
		\label{generation_cons2}
	\end{equation}
	
	\subsection{Problem Formulation and Decomposition}
	Our main objective is to minimize the system time delay. We formulate the problem as follows:
	\begin{align}
		\mathcal{P}: \quad &\underset{\boldsymbol{\alpha},\boldsymbol{l}, \boldsymbol{\phi},\boldsymbol{b}}{\min} \; \max\;  T_n \nonumber\\
		\text {s.t. }&E_n^{cp} + E_n^{mu} \leq \bar{E},\forall n \in \mathcal{N}^t,\label{energy_cons}\\ &(\ref{alpha_cons}),(\ref{bandwidth_cons1}),(\ref{bandwidth_cons2}),(\ref{phi_cons}),(\ref{generation_cons1})\text{ and }(\ref{generation_cons2}). \nonumber
	\end{align}
	Constraint (\ref{energy_cons}) guarantees the energy consumption satisfies maximum limit, and (\ref{alpha_cons}), (\ref{bandwidth_cons1}), (\ref{bandwidth_cons2}), (\ref{phi_cons}),  (\ref{generation_cons1}) and  (\ref{generation_cons2}) guarantee feasible decision.\par
	As we can see, $\mathcal{P}$ is a mixed-integer non-linear programming and obviously non-convex, which means $\mathcal{P}$ is NP-hard problem and it is very difficult to be directly solved. Moreover, $\mathcal{P}$ is a stochastic multi-timescale optimization problem. The problem is stochastic because decisions are influenced by the temporal dynamics of device computing capabilities and varying channel conditions during the training process. It is multi-timescale because decisions are made on two different time scales:
	\begin{itemize}
		\item Large Communication-scale (communication level): Decisions such as vehicle selection, made for all online vehicles, happen at a broader communication scale.
		\item Small Computation-scale (computation level): Once vehicles are selected, finer decisions regarding bandwidth allocation, transmission power, and data generation strategies are made for those selected vehicles, adapting to the computation requirements at that scale.
	\end{itemize}
	Additionally, the problem complexity is increased due to the coupling of vehicle selection, bandwidth allocation, transmission power assignment and data generation setting. To address these challenges, we propose a Two-Scale Algorithm. The large-scale algorithm focuses on label sharing and vehicle selection strategy, while the small-scale algorithm handles wireless resource allocation, uplink power management, and data generation strategy.
    {
        \begin{algorithm}[t]
        \caption{Bandwidth Allocation using Lagrange Multiplier Method} 
        \label{alg1}
        \begin{algorithmic}[1]
            \REQUIRE The initial Lagrange multipliers set ($\lambda_{1,n}^{(0)}$,$\lambda_{2}^{(0)}$,$\lambda_{3}^{(0)}$), the step size $\eta_{\lambda_1}$, $\eta_{\lambda_2}$ and $\eta_{\lambda_3}$;
            \REPEAT 
            \STATE Update the multiplier $\lambda_{1,n}^{(k+1)}$ as $\max\{\lambda_{1,n}^{(k)} + \eta_{\lambda_1} \frac{\partial L}{\partial \lambda_1}, 0\}$;
            \STATE Update the multiplier $\lambda_{2}^{(k+1)}$ as $\max\{\lambda_{2}^{(k)} + \eta_{\lambda_2} \frac{\partial L}{\partial \lambda_2}, 0\}$;
            \STATE Update the multiplier $\lambda_3^{(k+1)}$ as $\max\{\lambda_3^{(k)} + \eta_{\lambda_3} \frac{\partial L}{\partial \lambda_3}, 0\}$;
            \STATE Update the optimal $\boldsymbol{l}_n^{k+1}$ according to Eqn. (\ref{optimal_subp2}).
            \UNTIL{Convergence}
        \end{algorithmic}
    \end{algorithm} }       
\section{Proposed Scheme}
\label{section5}
	In this section, we firstly introduce the large communication-scale labels sharing and vehicle selection strategy, and then introduce the small computation-scale resource allocation and model augmentation. 
     {\begin{algorithm}[t]
        \caption{Transmission Power Assignment using SCA Method} 
        \label{alg2}
        \begin{algorithmic}[1]
        \REQUIRE Set of $\boldsymbol{l}_n$, $b^t$, $\alpha_n$, max energy constraints $\hat{E}$, the initial uplink power $\phi_n^0$ of vehicle $n$, iteration round $i=0$, the accuracy requirement $\epsilon$.
        \REPEAT
        \STATE calculate $\hat{e}(\phi_n^i,\phi_n)$ according to Eqn. (\ref{subp3-1});
        \STATE solve $\mathcal{SUBP}2$ by substituting $e(\phi_n^i)$ with $\hat{e}(\phi_n^i,\phi_n)$, substituting $t(\phi_n^i)$ with $\hat{t}(\phi_n^i,\phi_n)$, and achieve the optimal solution $\phi_n^{i,*}$
        \STATE $\phi_n \rightarrow \phi_n^{i,*}$,$i \rightarrow i+1$
        \UNTIL{$\left\|\phi_n^i-\phi_n^{i-1}\right\| \leq \epsilon$}
        \ENSURE	Optimal transmission power $\phi_n^*$.
        \end{algorithmic}
        \end{algorithm}} 
    \subsection{Large Communication-scale Labels Sharing and Vehicle Selection}
	We refer to the RSU communication range as the large communication scale. The process is structured in two phases: initially, we implement the label-sharing strategy, followed by the vehicle selection strategy.
	\subsubsection{Labels Sharing Strategy} 
	During the label-sharing stage, vehicles within the RSU's communication range send their low-privacy information along with the corresponding labels and data quality metrics $EMD_n^t$  of their local datasets, denoted as $EMD_n^t = \sum_{i=1}^{Y}\|p_n^t(y=i)-p(y=i)\|$. Due to the small size of the labels, the transmission delay is considerably shorter compared to that of the models. This efficiency enables interaction with a larger number of vehicles, thus obtaining a more diverse set of data labels.
        \subsubsection{Vehicle Selection Strategy}\label{section-vehicleselection} In this subsection, we address the vehicle selection subproblem. The vehicle selection problem is formulated as:
	\begin{align}
		\label{SUBP1} \mathcal{SUBP}1:& \underset{\boldsymbol{\alpha}}{\min \max}  \quad T_n \nonumber\\
		\text{s.t. }&(\ref{alpha_cons}).
	\end{align}
	As shown in Fig. \ref{fig:vehicle}, we assume that the arrival of vehicles within each RSU range follows the widely used Poisson distribution, and that the vehicles' average speed is related to the degree of curvature on the road \cite{xing2022secure,zheng2024mobility}. We thus have the average speed (km/h) $\Bar{v}_m$ of $\mathcal{N}^t$ vehicles in the service range of RSU as
	\begin{equation}
		\bar{v}_n =  \max{\{v_{max}(1-\frac{M}{M_{max}}),v_{min}\}},
	\end{equation}
	where $v_{max}$ is the maximum vehicle speed that can be driven within the server range of RSU. We assume roads in the RSU service range are uniform and have the same permissible maximum vehicle speed. Similarly, $v_{min}$ is the vehicle speed when the road is congested. Further, $M_{max}$ is the maximum allowable number of vehicles in RSU's server range on the road. In the case of free-flow traffic conditions, the speed of a vehicle $n$ in the service range of RSU, $v_{n}$ is a normally distributed random variable with the probability density function given by $f(v_n) = \frac{1}{\sqrt{2 \pi \sigma}} e^{-\frac{(v_n-\bar{v}_n)}{2\sigma^2}}$, where $\sigma = k\bar{v}_m$ and $v_{min} = \bar{v}_m - l\bar{v}_m$. The parameters $k,l$ is subject to the traffic activity observed in real-time.\par
     
	We define $\boldsymbol{v}_n$ as velocity of vehicle $n$ and $s_n$ as the remaining distance of the vehicle $n$ before running out of the coverage of the RSU \cite{10007714}, which is formulated by
	\begin{equation}
		s_n = \sqrt{r^2 - e^2} - \frac{\boldsymbol{v}_n}{v_n}x_n,
	\end{equation}
	where $r,e$ and $\sqrt{r^2-e^2}$ respectively represent the RSU cover radius, the vertical distance of the RSU to the road, and half of the coverage length of the RSU on the road. $\frac{v_n}{|v_n|}$ denote the running direction of the vehicle $n$. Thus, we have $t_n^{hold}$ the holding time that the vehicle $n$ can hold the V2R communication between it and the RSU before the vehicle runs out of the coverage of the RSU, which is given by 
	\begin{equation}
		t_n^{hold} = \frac{s_n}{v_n}.
	\end{equation} 		
	The vehicles on the road are at varying speeds and the wireless conditions are highly dynamic. Therefore, the diverse holding time for vehicles in the RSU’s coverage makes a great impact on the FL training accuracy. To ensure that selected vehicles can complete the task before leaving the RSU's coverage, we have
	\begin{align}
		\bar{T}_n = \min(t_n^{hold}, t^{max}),
	\end{align}
	where $t^{max}$ is the max time of training one round. And the time constraint is
	\begin{equation}
		T_n^{cp}+T_n^{mu} \leq \bar{T}_n, \forall n \in \mathcal{N}.
		\label{time_cons}
	\end{equation}
	With smaller $EMD$, with higher data quality\cite{zhao2018federated}. To avoid the high data heterogeneity, we also consider to select vehicles those have smaller $EMD$ value:
	\begin{equation}
		EMD_n^t \leq \hat{EMD}, \forall n \in \mathcal{N},
		\label{emd_cons}
	\end{equation}
	where $\bar{EMD}$ is the maximum tolerance for data quality.
	Involving all possible vehicles in the following training, we selected vehicles using the constraints Eqn. (\ref{time_cons}) and Eqn. (\ref{emd_cons}). So, the selected indicator is
	\begin{equation}
		\boldsymbol{\alpha^t} = \{ \alpha_n^t = 1 | (\ref{time_cons}) \wedge (\ref{emd_cons})\}.
	\end{equation}
	After selecting the vehicles, we implement the small computation-scale resource allocation and model augmentation algorithm.
       
   	{\begin{algorithm}[t]
            \caption{Joint Two-Scale Algorithm} 
            \label{alg-twoscale}
            \begin{algorithmic}[1]
                \REQUIRE Set training round $t=0$, and initialize $\epsilon_1,\epsilon_2,\epsilon_3>0$.
                \REPEAT
                \STATE \textbf{Large Communication Scale:}
                \STATE For each vehicle $n \in \mathcal{N}$, share labels and their corresponding $EMD_n^t$.
                \STATE Select vehicles by solving $\mathcal{SUBP}$1.
                \STATE \textbf{Small Computation Scale (Iteration Algorithm):}
                \STATE For each vehicle $n \in \mathcal{N}^t$, set BCD iteration $i=0$, and initialize $\boldsymbol{l}_n^{0}, \phi_n^{cm,0}, b_g^{0}$. 
                \REPEAT 
                \STATE $i =  i + 1$;
                \STATE Compute the optimal bandwidth allocation $\boldsymbol{l}_n^{i}$ given $\phi_n^{cm,i-1}$ and $b_g^{i-1}$ for all $n \in \mathcal{N}^t$ by solving $\mathcal{SUBP}$2 in Algorithm \ref{alg1}.
                \STATE Choose the optimal transmission power $\phi_n^{cm,i}$ given $b_g^{i-1}$ and $\boldsymbol{l}_n^{i}$ by solving $\mathcal{SUBP}$3 in Algorithm \ref{alg2}.
                \STATE Calculate the optimal number of generated images $b_g^{i}$ given $\phi_n^{cm,i}$ and $\boldsymbol{l}_n^{i}$ using Eqn. (\ref{eq-bgt}) of $\mathcal{SUBP}$4.
                \UNTIL{$\Vert \boldsymbol{l}_n^{i}-\boldsymbol{l}_n^{i-1}\Vert<\epsilon_1$, $\Vert \phi_n^{i}-\phi_n^{i-1}\Vert<\epsilon_2$, and $\Vert b_g^{i}-b_g^{i-1}\Vert<\epsilon_3$}
                \STATE \textbf{Model training and augmentation:}
                \FOR{each $n \in \mathcal{N}^t$}
                \STATE Vehicle $n$ trains its local model on the local dataset, then uploads the local model to the RSU.
                \ENDFOR
                \STATE While vehicles are training local models, the RSU generates images. Upon receiving the local models, the RSU aggregates them and uses the generated images to augment the model using Eqn. (\ref{modelaugmentation}).
                \UNTIL{Model Convergence}
            \end{algorithmic}
            \end{algorithm}}
    \subsection{Small Computation-scale Resource Allocation and Model Augmentation}
        \subsubsection{Problem Transformation}
    	To simplify the solution process and avoid the complexity of the min-max problem, we introduce a latency tolerance $\bar{T}$ as an upper bound, represented as
    	\begin{equation}
    		\bar{T} =  \max_{n=1}^{\mathcal{N}^t} \{T_n^{cp} + T_n^{mu}\}.
    		\label{barT}
    	\end{equation}
    	The small computation-scale problem $\mathcal{P'}$ is defined as :
    	\begin{align}
    		\mathcal{P'}:\quad & \underset{\boldsymbol{l}, \boldsymbol{b},
    			\boldsymbol{\phi}}{\min} \quad \bar{T}\nonumber\\
    		\text{s.t. }&T_n^{cp}+T_n^{mu} \leq \bar{T},\forall n \in \mathcal{N}^t,\label{bartime_cons}\\
    		&(\ref{alpha_cons}),(\ref{bandwidth_cons1}),(\ref{bandwidth_cons2}),(\ref{phi_cons}),(\ref{generation_cons1}),(\ref{generation_cons2})\text{and}(\ref{energy_cons}). \nonumber
    	\end{align}
    	To solve $\mathcal{P}'$, we decompose $\mathcal{P}'$ into three subproblem, optimal bandwidth allocation problem  $\mathcal{SUBP}2$, optimal transmission power assignment problem $\mathcal{SUBP}3$ and data generation problem $\mathcal{SUBP}4$.
        \subsubsection{Optimal Bandwidth Allocation}
    	For simplicity, we reformulate (\ref{bartime_cons2}) to replace (\ref{bartime_cons}) and substitute  (\ref{energy_cons2}) for (\ref{energy_cons}). Consequently, the subproblem $\mathcal{SUBP}2$ can be expressed as follows:
    	\begin{align}
    		\mathcal{SUBP}2:&\underset{\{l_n\}_{n \in \mathcal{N}^t}}{\min} \bar{T} \nonumber\\
    		\text {s.t. }&(\ref{bandwidth_cons1}),(\ref{bandwidth_cons2}),\nonumber\\& A + \frac{B}{l_n}\leq \bar{T}, \forall n \in \mathcal{N}^t,\label{bartime_cons2}\\
    		& C+\frac{D}{l_n} \leq \bar{E}, \forall n \in \mathcal{N}^t.\label{energy_cons2}
    	\end{align}
        Here, to enhance readability, we introduce the notations $A, B, C,$ and $D$ to replace complex expressions: $A = t_n^0  + \frac{c_1 b_n \theta_n^{mem}}{{f_n^{mem}}} \nonumber  + \frac{c_2 b_n \theta_n^{core}}{{f_n^{core}}}$, $B = \frac{s(\boldsymbol{\omega})}{W log_2 (1+\frac{\phi_n h_0 d_n^{-\gamma}}{N_0})}$, $C =  \{p_n^{G0} + \zeta_n^{mem} {f_n^{mem}} + \zeta_n^{core} (V_n^{core})^2 {f_n^{core}}\} \cdot \{t_n^0 + \frac{c_1 b_n \theta_n^{mem}}{{f_n^{mem}}} + \frac{c_2 b_n \theta_n^{core}}{{f_n^{core}}}\}$ and $D = \phi_n \frac{s(\boldsymbol{\omega})}{W log_2 (1+\frac{\phi_n h_0 d_n^{-\gamma}}{N_0})}$.\par
        The subproblem $\mathcal{SUBP}2$ is integer linear programming (ILP) problem, which is NP-hard. As mentioned before, there are only $M$ subcarriers for selected $N$ vehicles. The index of vehicles which is assigned the subcarrier $n$ at the $t$-th training round is denoted as $\mathbf{I}(n)$. Then, the expected number of subcarriers $l_n$ can be expressed as:
    	\begin{equation}
    		\boldsymbol{l}_n = \mathbb{E} \left[ \sum_{n=1}^{N}\mathbf{1}(n \in \mathbf{I}(n)) \right]
    	\end{equation}
    	Clearly, if $\alpha_n^t=0$, i.e. vehicle $n$ is not selected in round $t$, then it is preferred not to allocate any bandwidth to this vehicle. On the other hand, if $\alpha_n^t=1$, then we require that at least a minimum bandwidth $\boldsymbol{l}_{min}$ is allocated to vehicle $n$. In addition, a close-to-zero bandwidth allocation will require an extremely high transmit power and hence result in an extremely high energy consumption to achieve a target transmission rate\cite{9237168}. To make it possible, we assume $\boldsymbol{l}_{min} \leq \frac{1}{N_t}$.\par
	Applying Lagrange multiplier method, we can derive
    	\begin{align}
    		L = \bar{T}+&\sum_{n \in \mathcal{N}^t}\lambda_{1,n} (A+\frac{B}{\boldsymbol{l}_n}-\bar{T}) \nonumber \\+& \lambda_{2} \sum_{n \in \mathcal{N}^t} (C+\frac{D}{\boldsymbol{l}_n}-\bar{E})+\lambda_{3}(\sum_{n \in \mathcal{N}^t}\boldsymbol{l}_n-W),
    	\end{align}
    	where $\lambda_{1,n}$, $\lambda_{2,n}$ and $\lambda_3$ are the Lagrange multipliers related to constraints C1 and C3.\par
	   Applying KKT conditions, we can derive the necessary and sufficient conditions in the following,
    	\begin{equation}
    		\frac{\partial L}{\partial \boldsymbol{l}_n} =  \lambda_{1,n} (-\frac{B_n}{\boldsymbol{l}_n^2}) + \lambda_{2}(-\frac{D_n}{\boldsymbol{l}_n^2})+\lambda_3
    	\end{equation}
	Finally, the optimal frequency of $\boldsymbol{l}_n^*$ is 
    	\begin{equation}
    		\boldsymbol{l}_n^* = \sqrt{\frac{\lambda_{1,n} B_n + \lambda_{2} D_n}{\lambda_3}}.
    		\label{optimal_subp2}
    	\end{equation}
	The algorithm is presented in Algorithm \ref{alg1}.
            \begin{figure}[t]
                \centering
                \includegraphics[width=0.90\linewidth]{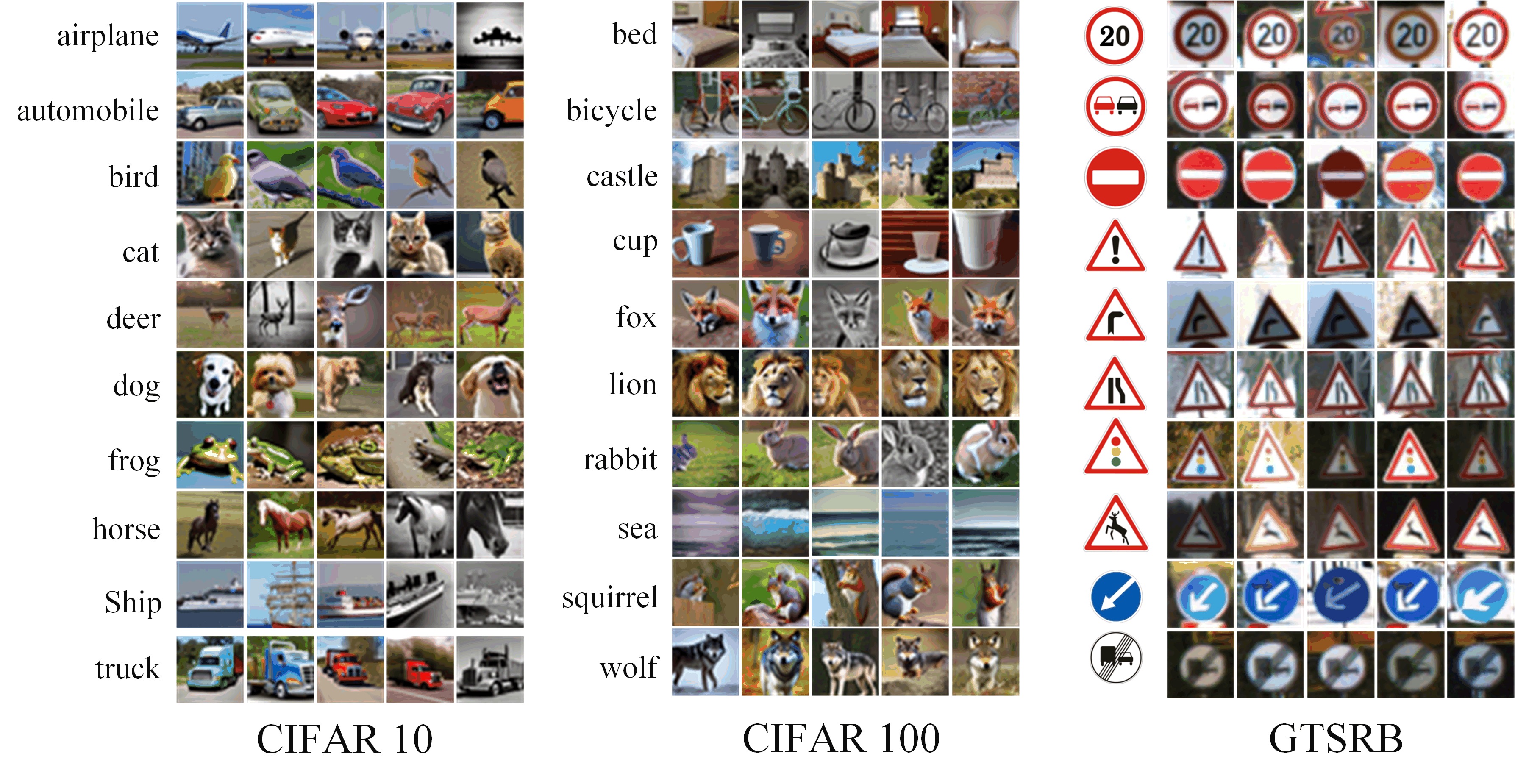}
                \caption{Some generated images.}
                \label{fig:some_generated_images}
            \end{figure}   

        \subsubsection{Optimal Transmission Power Assignment}
        For simplicity, we reformulate (\ref{bartime_cons3}) to replace (\ref{bartime_cons}) and substitute constraint (\ref{energy_cons3}) for (\ref{energy_cons}). Consequently, the subproblem $\mathcal{SUBP}3$ can be expressed as follows:
    	\begin{align}
    		\mathcal{SUBP}3:\quad&\underset{\boldsymbol{\phi}}{\min}\quad \bar{T} \nonumber \\
    		\text {s.t. }
    		&(\ref{phi_cons}),\nonumber\\
    		& A + \frac{s(\boldsymbol{\omega})}{Wlog_2(1+\frac{{ \phi_n} h_0d_n^{-\gamma}}{N_0})}<=\bar{T}, \forall n \in \mathcal{N}^t,\label{bartime_cons3}\\
    		& G+{\phi_n} \frac{s(\boldsymbol{\omega})}{W log_2 (1+\frac{ { \phi_n} h_0 d_n^{-\gamma}}{N_0})} \leq \bar{E}, \forall n \in \mathcal{N}^t, \label{energy_cons3}
    	\end{align}
    	where $G = \{ p_n^{G0} + \zeta_n^{mem} f_n^{mem} + \zeta_n^{core} (V_n^{core})^2 f_n^{mem}\} \cdot \{t_n^0 + \frac{c_1 b_n \theta_n^{mem}}{f_n^{mem}} + \frac{c_2 b_n \theta_n^{core}}{f_n^{mem}}\}$.\par
    	It can be observed that the constraints (\ref{bartime_cons3}) and (\ref{energy_cons3}) is non-convex for the upload power $\phi_n$, which makes $\mathcal{SUBP}3$ a uni-variate non-convex optimization problem.\par
    	We define $\phi_n^{i}$ as the upload power of vehicle $n$ in $i$-th iteration. We denote the non-convex part of (\ref{bartime_cons3}) as 
    	\begin{equation}
    		t(\phi_n)= \frac{s(\omega)}{Wlog_2(1+\frac{\phi_nh_0d_n^{-\gamma}}{N_0})}.
    	\end{equation}
    	To obtain the approximate upper bound, $t(\phi_n)$ can be approximated by its first-order Taylor expansion $\hat{t}(\phi_n^{i},\phi_n)$ at point $p_n^{i}$, which is given by:
    	\begin{equation}
    		\hat{t}(\phi_n) = t(\phi_n^{i}) + t'(\phi_n^{i})(\phi_n-\phi_n^{i}),
    	\end{equation}
    	where $t'(\phi_n^{i})$ is denoted as the first-order derivative of $t(\phi_n^{i})$ at point $\phi_n^{i}$:
    	\begin{align}
    		t'(\phi_n^{i}) = \frac{-A'B'\ln2}{(1+B'\phi_n^i)(\ln(1+B'\phi_n^i))^2}.
    	\end{align}
    	where $A'=\frac{s(\omega)}{w_n}$ and $B'=\frac{h_0 d_n^{-\gamma}}{N_0}$. The problem is convex at each SCA iteration by changing $t(\phi_n)$ to $\hat{t}(\phi_n^i,\phi_n)$ in $\mathcal{SUBP}3$. Then Eqn. (\ref{bartime_cons3}) can change into $ \hat{t} \leq \bar{T}$ to be convex. \par
	   Same to (\ref{energy_cons3}), we denote the non-convex part of (\ref{energy_cons3}) as 
    	\begin{equation}
    		e(\phi_n)= \phi_n \times \frac{s(\omega)}{W_n log_2(1+\frac{\phi_n h_0 d_n^{-\gamma}}{N_0})}.
    	\end{equation} 
    	To obtain the approximate upper bound, $e(\phi_n)$ can be approximated by its first-order Taylor expansion $\hat{e}(\phi_n^{i},\phi_n)$ at point $\phi_n^{i}$, which is given by:
    	\begin{equation}
    		\hat{e}(\phi_n) = e(\phi_n^{i}) + e'(\phi_n^{i})(\phi_n-\phi_n^{i}),
    		\label{subp3-1}
    	\end{equation}
    	where $e'(\phi_n^{i})$ is denoted as the first-order derivative of $e(\phi_n^{i})$ at point $\phi_n^{i}$:
    	\begin{align}
    		\label{subp3-2}
    		e'(\phi_n^{i}) = &\frac{A'}{\log_2(1+B'\phi_n^{i})} \nonumber \\
    		-&\frac{A'B'\phi_n^{i}}{\ln2 (1+B'\phi_n^{i}) (\log_2(1+B'\phi_n^{i}))^2}.
    	\end{align}
    	The problem is convex at each SCA iteration by changing $e(\phi_n)$ to $\hat{e}(\phi_n^i,\phi_n)$ in $\mathcal{SUBP}3$. Then (\ref{energy_cons3}) can change into $ \hat{e} \leq \hat{E}$ to be convex. So we can obtain the optimal uplink power $\phi_n^{i,*}$ using Algorithm \ref{alg2}. 
	    \begin{figure*}[t]
                \centering
                \subfloat[CIFAR10]{
                    \includegraphics[width=0.32\textwidth]{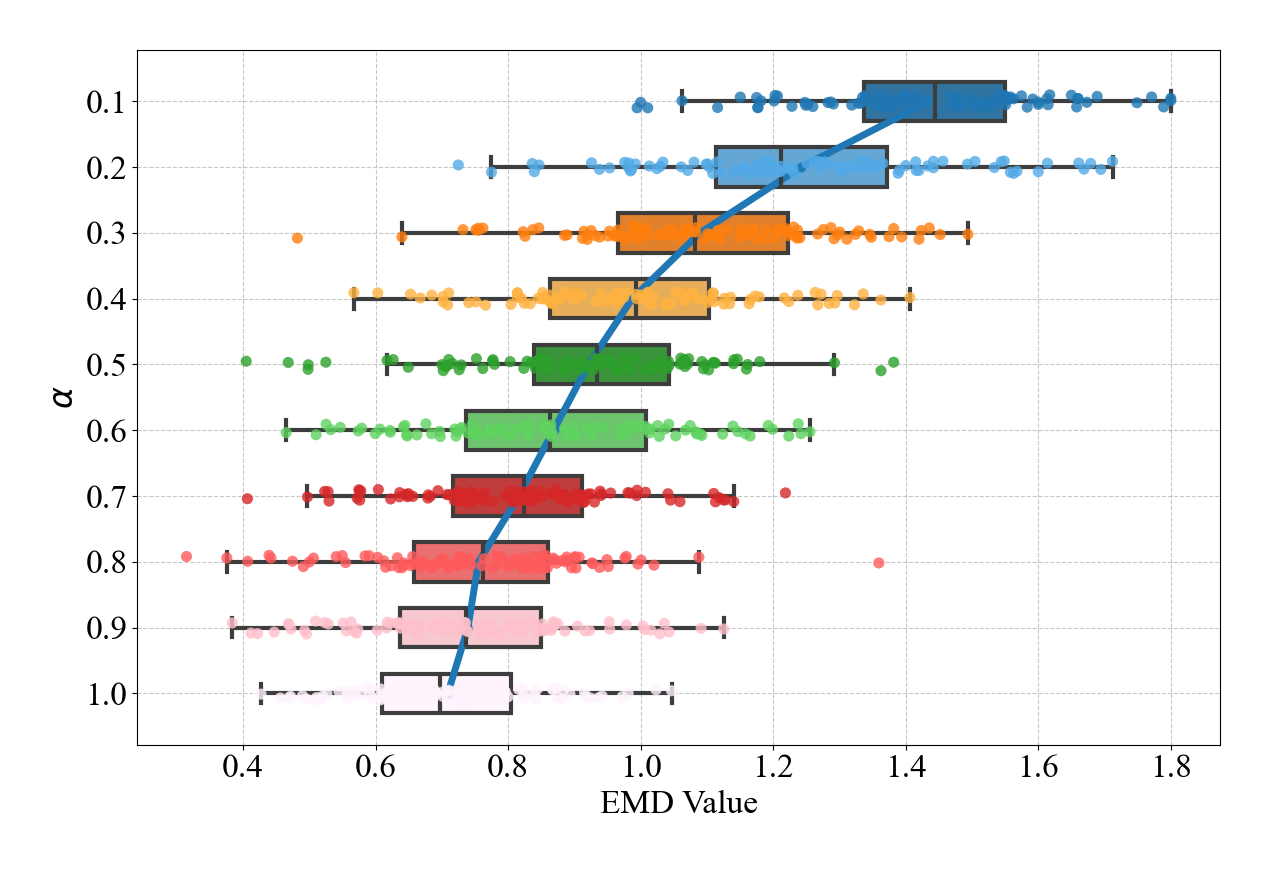}
                }
                \subfloat[CIFAR100]{
                    \includegraphics[width=0.32\textwidth]{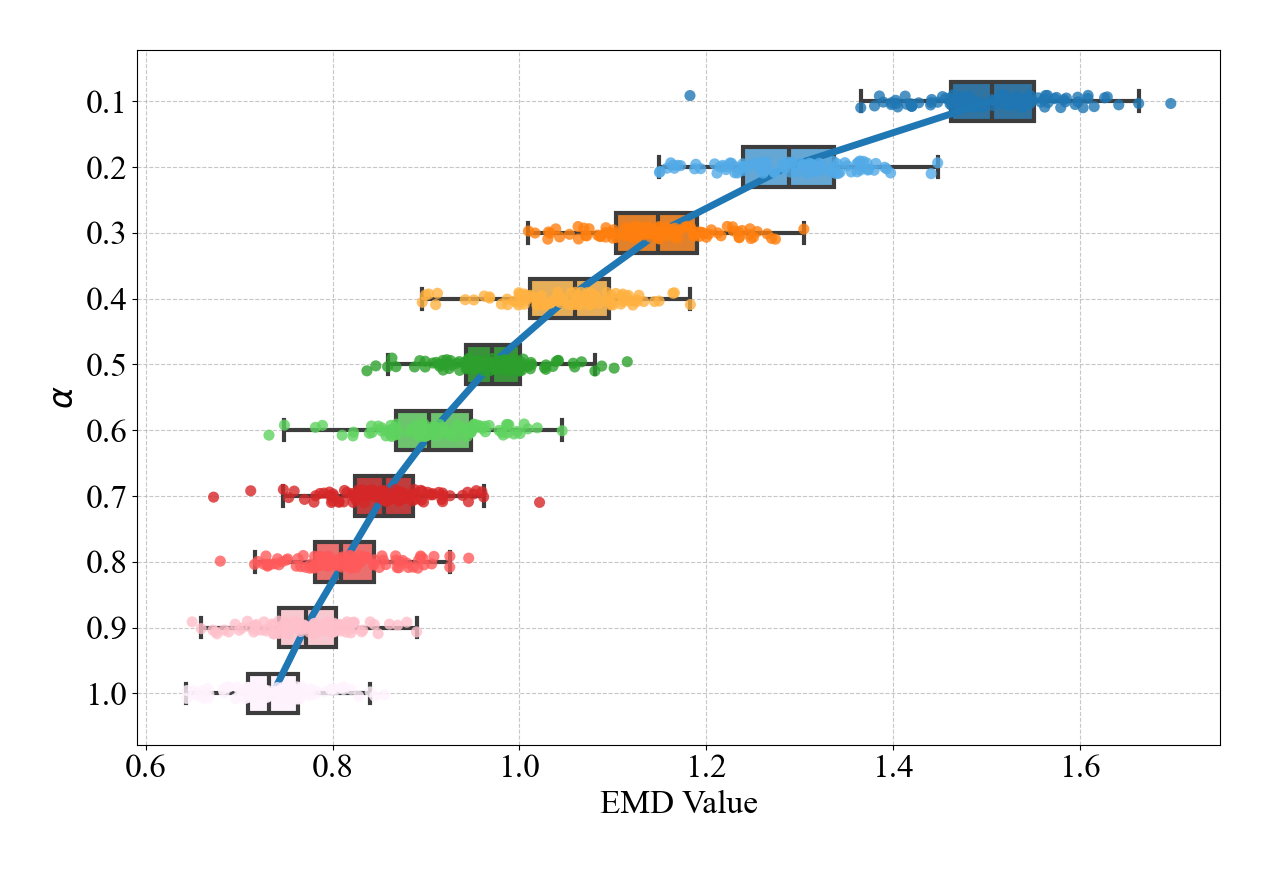}
                }
                \subfloat[GTSRB]{
                    \includegraphics[width=0.32\textwidth]{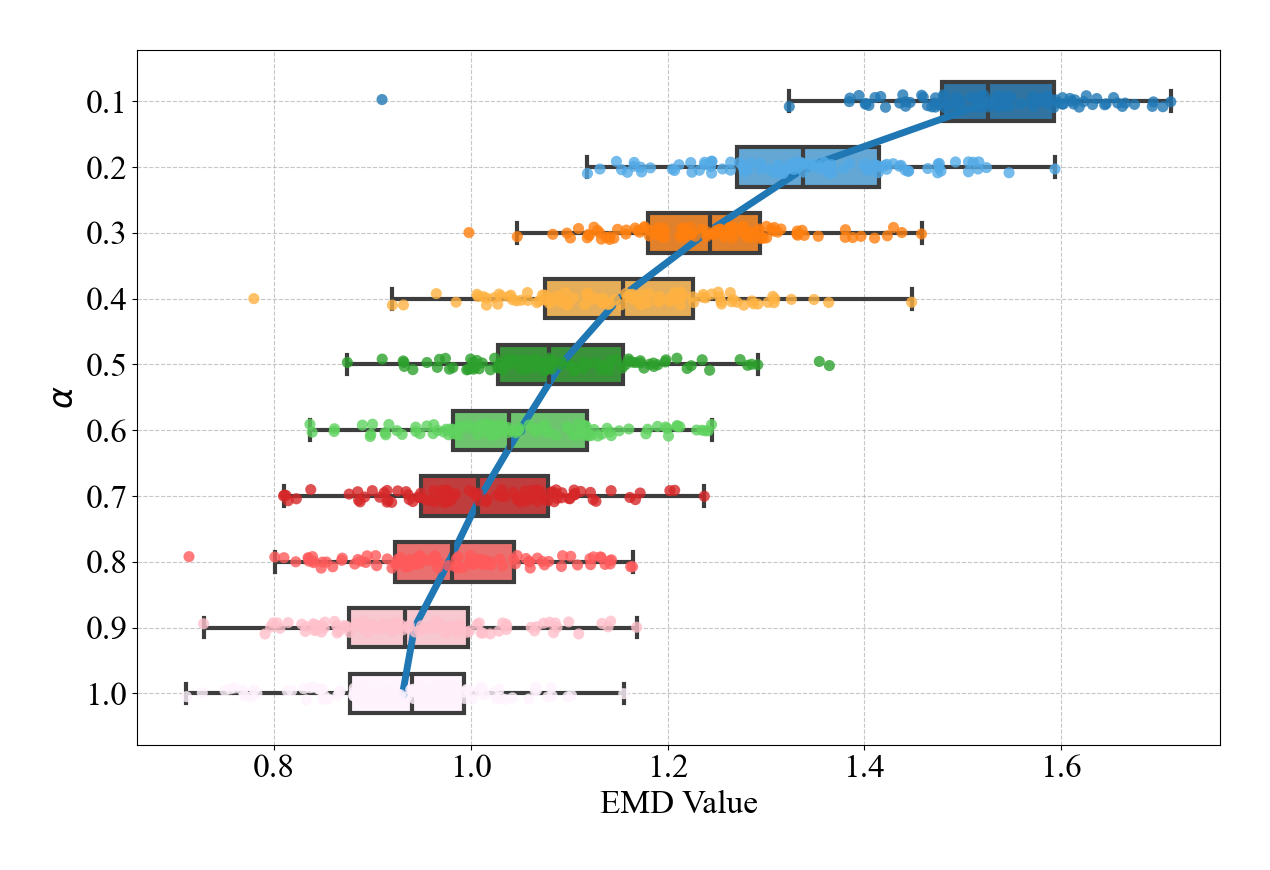}
                }
                \caption{EMD value with varying Dirichlet distribution $\alpha$.}
                \label{fig:emd}
            \end{figure*}
        \subsubsection{Data Generation Strategy}
    	\begin{align}
    		\mathcal{SUBP}4:\quad&\underset{\boldsymbol{b}}{\min}\quad \bar{T}\nonumber \\
    		&\text {s.t. } (\ref{generation_cons1}),(\ref{generation_cons2}).
    	\end{align}
    	According to ($\ref{barT}$), the maximum time $\bar{T} = max_{n \in \mathcal{N}^t}\{T_n^{cp}+T_n^{mu}\}$. Combine the conditions  (\ref{generation_cons1}) and ($\ref{generation_cons2}$), we can derive the optimal solution as
    	\begin{equation}
    		\lfloor\frac{\max_{n \in \mathcal{N}^t}{\{T_n^{cp}+T_n^{mu}\}}-T_s^{cp}(\boldsymbol{b}^{t-1})}{\sum_{i=1}^{I}\frac{d_{m,i}}{f_{rsu}}}\rfloor,
    		\label{eq-bgt}
    	\end{equation}
    	where $\lfloor \cdot \rfloor$ denotes the floor function.\par
    	To better utilize the generated images, we generate them based on an IID data distribution. This means that in each round, the server generates an equal number of images for the labels obtained through label sharing. The total number of generated images is calculated using equation (\ref{eq-bgt}).
   
    \subsection{Two-Scale Algorithm and Complexity Analysis}
	In this section, we first introduce the joint two-scale algorithm, and then analysis the complexity of the two-scale algorithm.
	\subsubsection{Joint Two-scale Algorithm}
	In Algorithm \ref{alg-twoscale}, $t$ is initially set to 0. First, the RSU applies the labels sharing and vehicle selection strategy in a large communication scale by solving $\mathcal{SUBP}1$. Based on the selected vehicles, $\mathcal{P}'$ is then solved in a smaller computation scale using the block coordinate descent (BCD) approach. The optimal bandwidth allocation $\boldsymbol{l}_n^{t,i}$ is obtained by fixing $\phi_n^{t,i-1}$ and $b_g^{t,i-1}$. The optimal transmission power $\phi_n^{t,i}$ is computed with $\boldsymbol{l}_n^{t,i}$ and $b_g^{t,i-1}$ given. Then, $b_g^{t}$ is directly calculated based on $\boldsymbol{l}_n^{t,i}$ and $\phi_n^{t,i}$. The loop continues until the differences meet the threshold requirements $\epsilon_1$, $\epsilon_2$, and $\epsilon_3$.

	\subsubsection{Complexity Analysis}
	The complexity of Algorithm \ref{alg-twoscale} mainly composed with four subproblems in two scale. In large-communication scale, the complexity of $\mathcal{SUBP}$1 is $\mathcal{O}(N)$. In small-computation scale, there are three subproblems using BCD method. The complexity of $\mathcal{SUBP}$2 using Lagrange Multiplier Method is $\mathcal{O}(M^{3.5}\log(\frac{1}{\epsilon_1}))$ according to\cite{bomze2010interior}, where $M$ is the number of variables. The complexity of $\mathcal{SUBP}$3 using SCA method is $\mathcal{O}(I_{SCA}M^3)$, where $I_{SCA}$ is the iteration of SCA. The complexity of $\mathcal{SUBP}4$ is $\mathcal{O}(1)$. Therefore, the overall computation complexity in Small-Computation scale is $\mathcal{O}(I_{BCD}M^{3.5}\log(\frac{1}{\epsilon_1}) + I_{BCD}I_{SCA}M^3)$, where $I_{BCD}$ is the number of iterations of BCD algorithm.	So, the joint Two-scale algorithm is $\mathcal{O}(N+I_{BCD}M^{3.5}\log(\frac{1}{\epsilon_1})+I_{BCD}I_{SCA}M^3)$.

\section{Performance Evaluation}
\label{section6}
In this section, we evaluate the performance of the proposed GenFV scheme and resource management algorithm through comprehensive simulations.

\subsection{Simulation Setup}
\subsubsection{Setting for FL Training and Datasets}
 We utilize the ResNet-18 model for the image classification task in FL, with a learning rate set to $1 \times 10^{-4}$ and a batch size of 64. Our simulation leverages three distinct image classification datasets: (1) the CIFAR-10 dataset \cite{krizhevsky2009learning}, containing colored images categorized into 10 classes such as "Airplane" and "Automobile". CIFAR-10 dataset comprises a training set with 50,000 samples and a testing set with 10,000 samples. (2) the CIFAR-100 dataset \cite{krizhevsky2009learning}, containing colored images categorized into 100 classes such as "Apple" and "Dolphin". Each class has exactly 500 training samples and 100 test samples, making a total of 50,000 training samples and 10,000 test samples for performance evaluation. (3) the GTSRB dataset contains 43 classes of traffic signs, divided into 39,209 training images and 12,630 test images. These images feature various lighting conditions and diverse backgrounds \cite{Houben-IJCNN-2013}. \par
 Notably, data distribution at vehicles is non-IID, which widely exists in practical systems. We realize data heterogeneity using Dirichlet distribution with parameter $\alpha$, where lower $\alpha$ results into more heterogeneous data partitions. The number of vehicles in the coverage follows a Possion Distribution. 

        \begin{figure}[t]
		\centering
		\subfloat[Training Loss]{\includegraphics[width=0.40\textwidth]{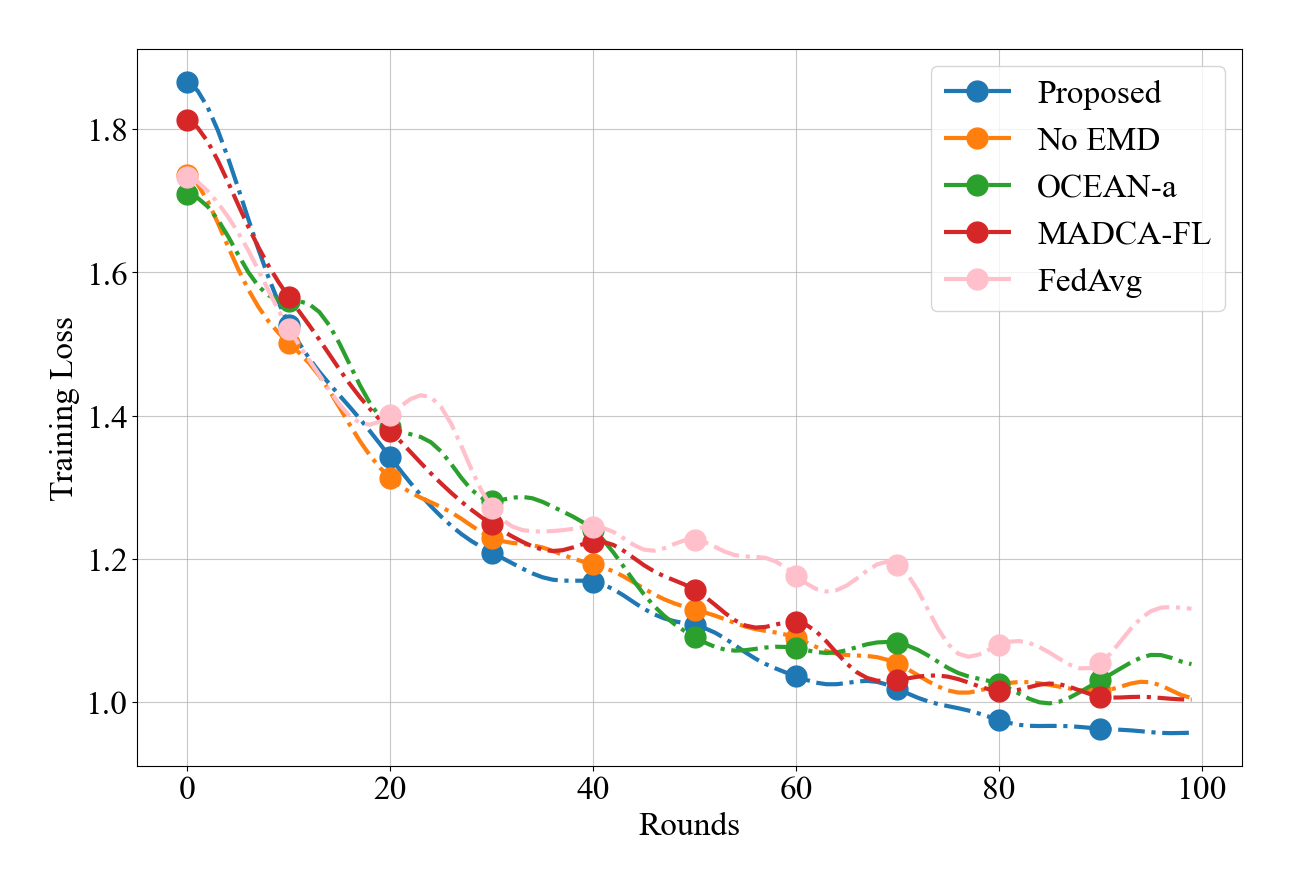}\label{trainloss}}\hfill
		\subfloat[Testing Accuracy]{\includegraphics[width=0.40\textwidth]{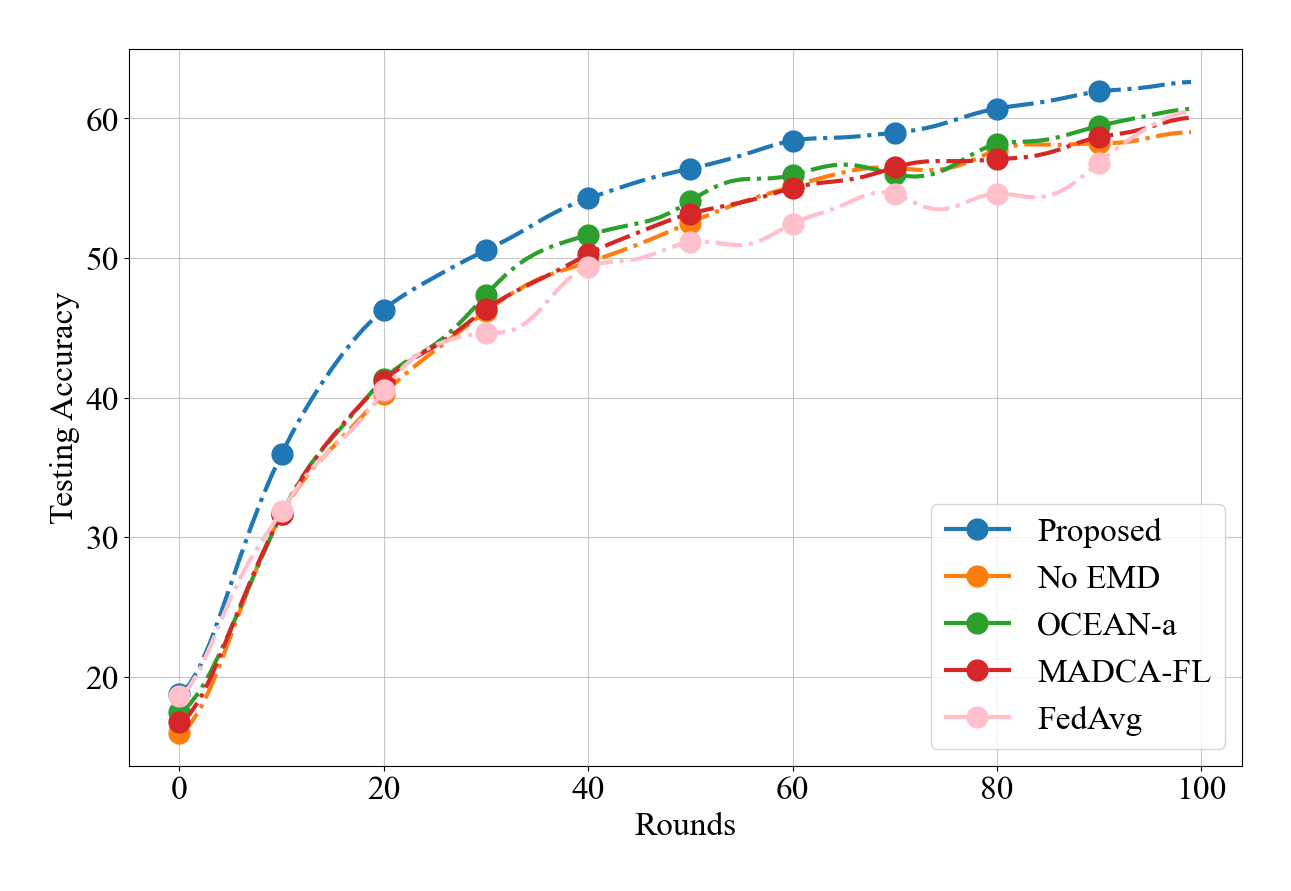}\label{testacc}}
		\caption{Training Loss and Testing Accuracy on CIFAR-10 with different vehicle selection strategy.}
		\label{fig:cifar10-selection-compare}
	\end{figure}	
  
	\begin{figure*}[t]
		\centering
		\begin{minipage}[t]{0.32\linewidth}  
			\centering
			\includegraphics[width=1.0\textwidth]{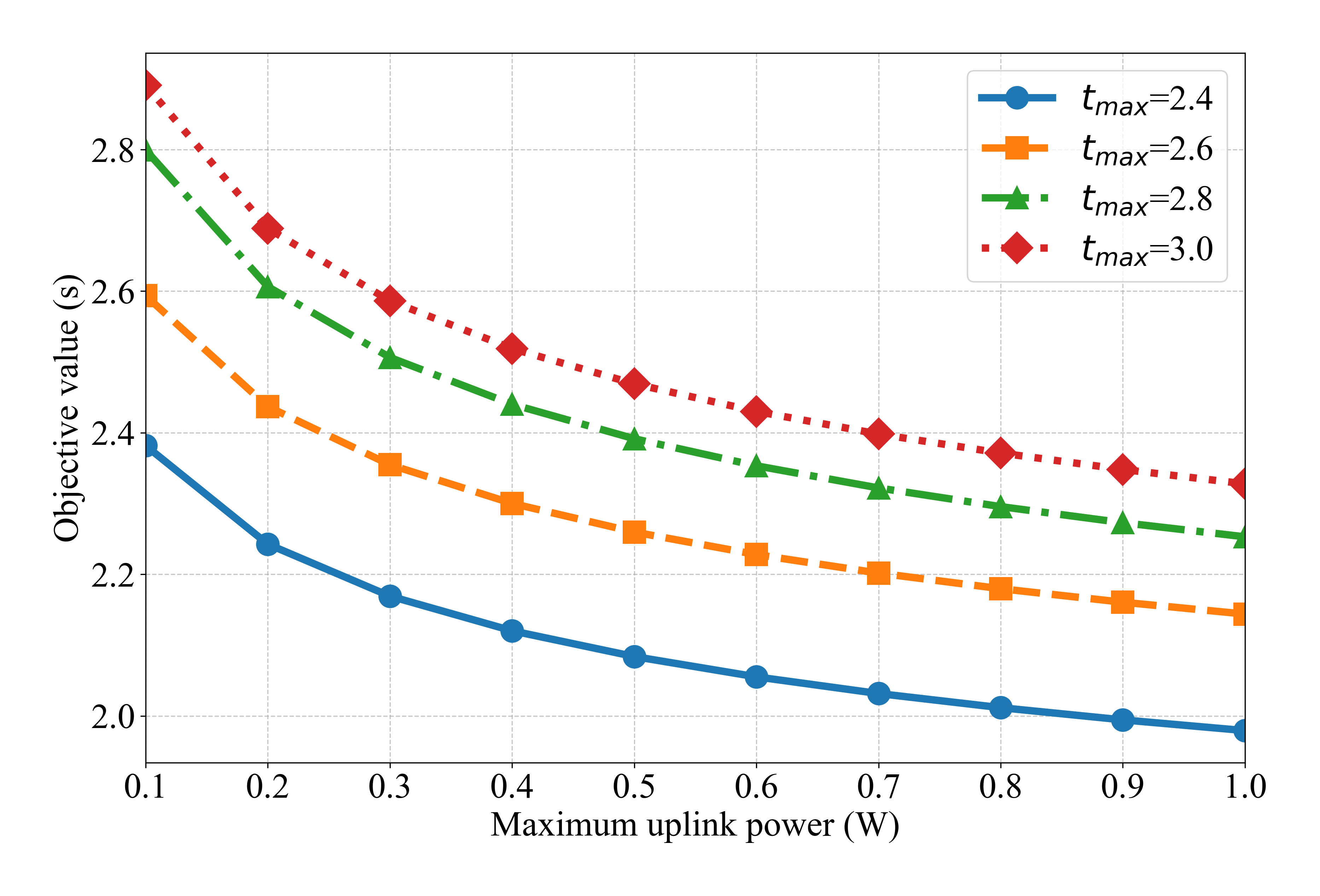} 
			\caption{Objective Value vs. Maximum Uplink Power under different $t_{max}$ constraint.}
			\label{fig:convergencetmax} 
		\end{minipage}
		\begin{minipage}[t]{0.32\linewidth}
			\centering
			\includegraphics[width=1.0\textwidth]{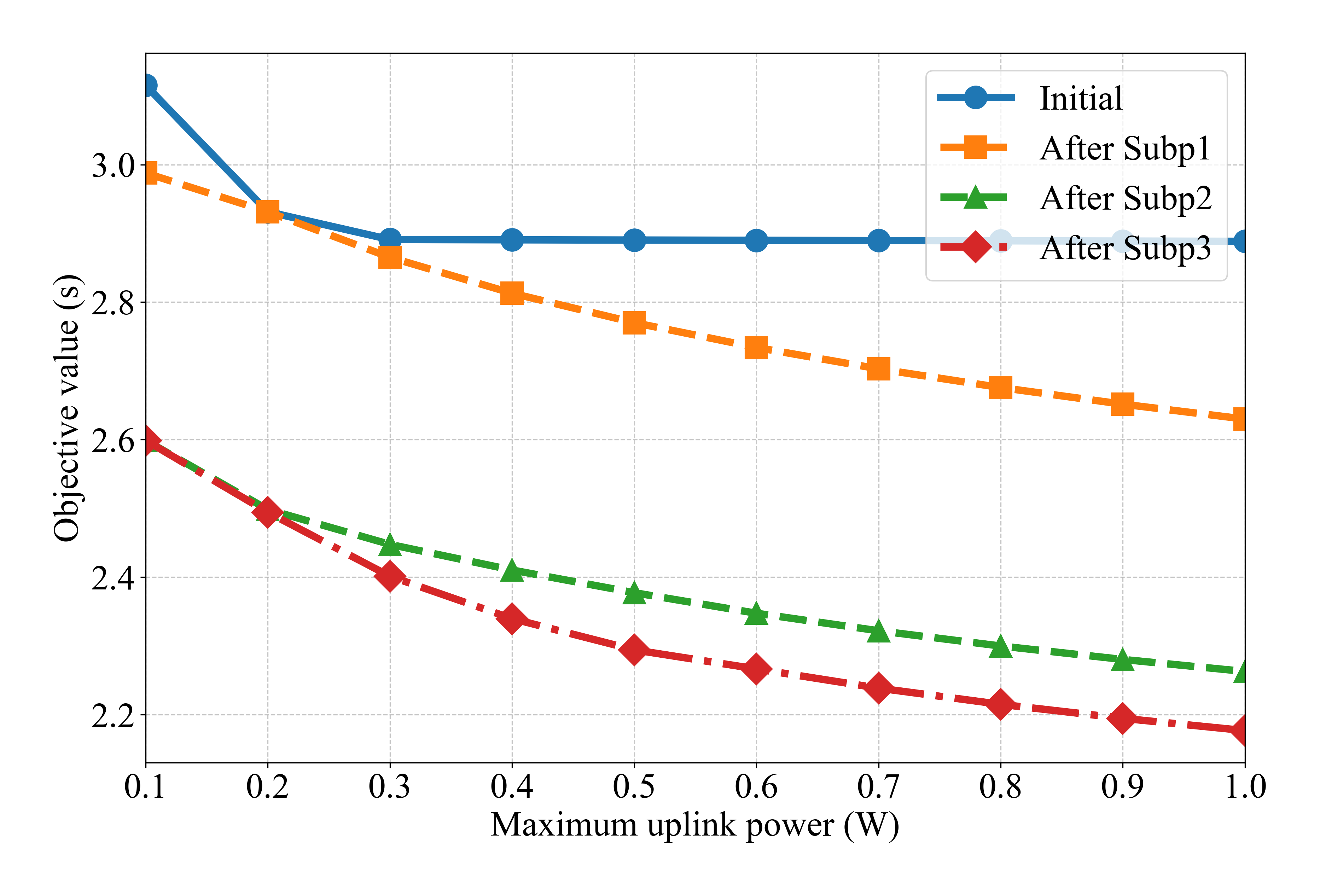} 
			\caption{Objective Value vs. Maximum Uplink Power under different selection and optimization strategy.}
			\label{fig:convergenceselection} 
		\end{minipage}
		\begin{minipage}[t]{0.32\linewidth}
			\centering
			\includegraphics[width=1.0\textwidth]{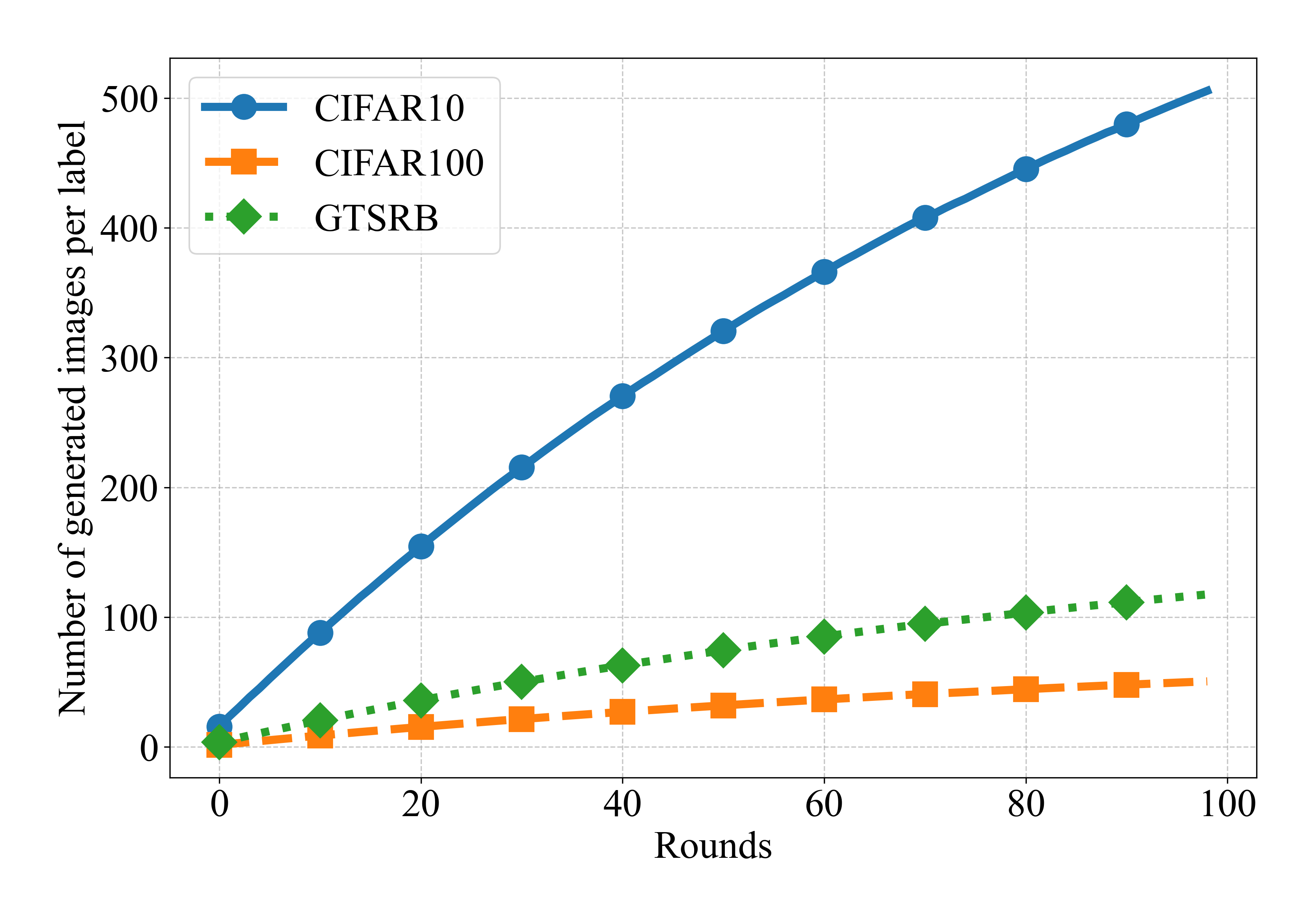} 
			\caption{The mount of generated images with different dataset.}
			\label{fig:generated_images} 
		\end{minipage}
        \end{figure*}

\subsubsection{Setting for Images Generation}
We employ a diffusion model\cite{ho2020denoising,10557146,li2024filling} for image generation, where generating a batch of 64 images sized 32x32 under default setting. Some of the generated images is shown in Fig. \ref{fig:some_generated_images}, which are visually good.
\subsubsection{Setting for Communication and Computation Model}
The GPU memory frequency of vehicles is range from $1250$ MHz to $1750$ MHz and the GPU core frequency is range from $1000$ MHz to $1600$ MHz. The number of transmission power $\phi_n$ for vehicle $n$ to process one sample data is randomly setting from $0.1$W to $1$W. The transmission power of RSU is $40$ dBm. The noise power $\sigma_0^2$ is -174 dBm. There are $M = 20$ subcarriers to be assigned, and every sub channel bandwidth is $10$ MHz.
\subsubsection{Setting for \texorpdfstring{$\hat{EMD}$}{EMD} Value}
The EMD constraint (\ref{emd_cons}) plays a crucial role in effectively managing the heterogeneity of vehicle data and ensuring robust model performance across diverse data distributions. Fig. \ref{fig:emd} illustrates the EMD values for vehicles with varying Dirichlet distributions under different datasets. Based on these observations, we establish a threshold $\hat{EMD}$ to mitigate excessive heterogeneity among vehicles, with the specific values detailed in Table \ref{tab:simple_4x4}.
\begin{table}[ht]
    \centering
    \caption{The value of \texorpdfstring{$\hat{EMD}$}{EMD} with different \texorpdfstring{$\alpha$}{alpha} and datasets.}

    \begin{tabular}{lcccc}
        \toprule
        \textbf{Dataset} & \textbf{$\alpha=0.1$} & \textbf{$\alpha=0.3$} & \textbf{$\alpha=0.5$} & \textbf{$\alpha=1.0$} \\
        \midrule
        CIFAR10 & 1.5 & 1.2 & 1.0 & 0.8 \\
        CIFAR100 & 1.5 & 1.2 & 1.0 & 0.8 \\
        GTSRB & 1.5 & 1.3 & 1.2 & 1.0 \\
        \bottomrule
    \end{tabular}
    \label{tab:simple_4x4}
\end{table}

    \begin{figure*}[t]
    	\centering
    	\subfloat[Dir($\alpha=0.1$)]{\includegraphics[width=0.23\textwidth]{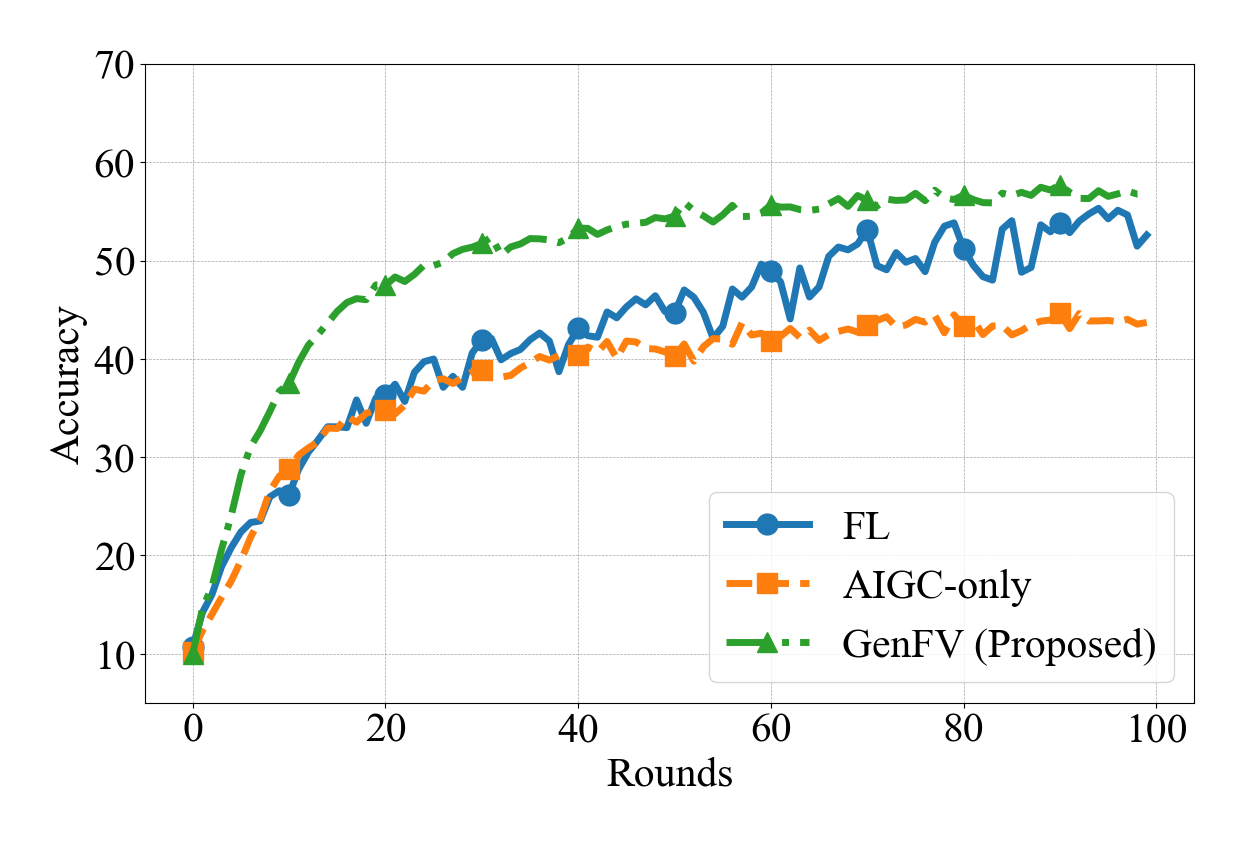}}
    	\hfill
    	\subfloat[Dir($\alpha=0.3$)]{\includegraphics[width=0.23\textwidth]{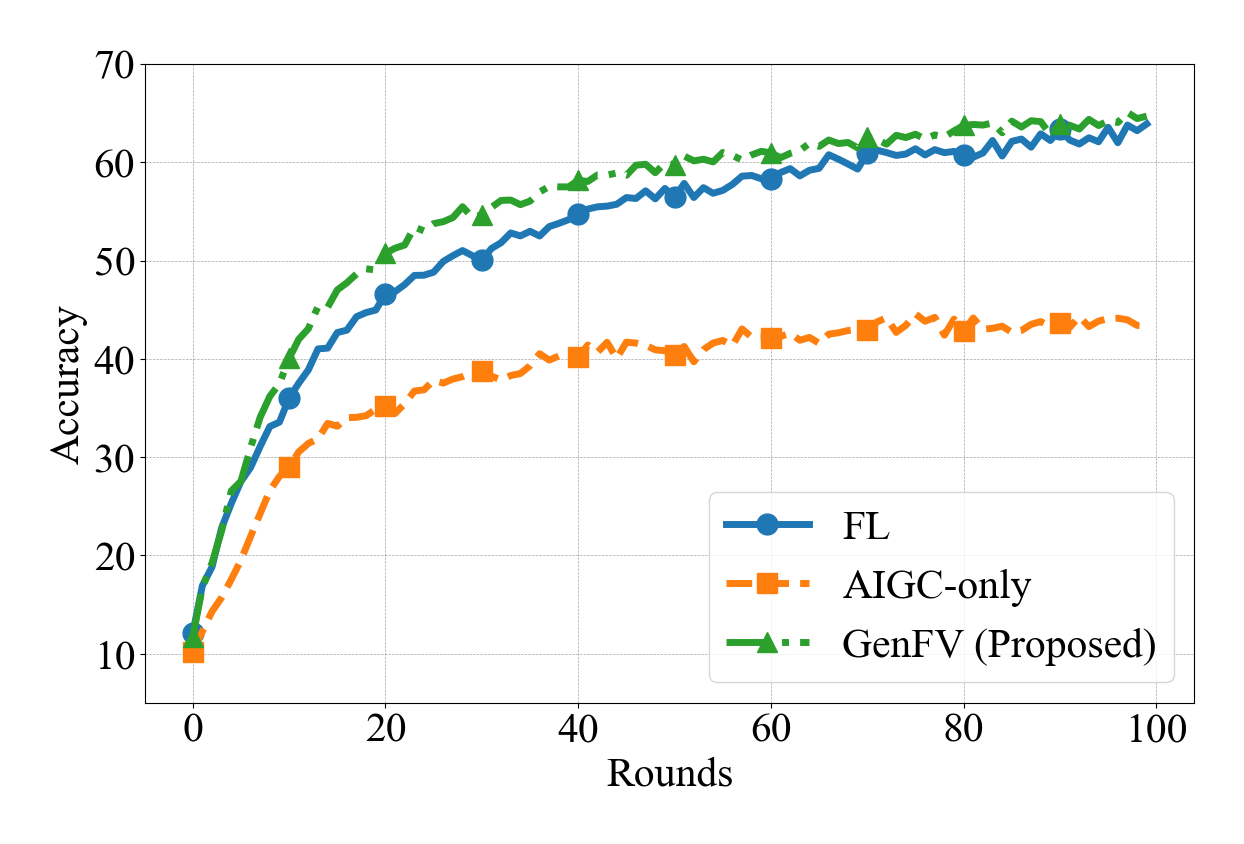}}
    	\hfill
    	\subfloat[Dir($\alpha=0.5$)]{\includegraphics[width=0.23\textwidth]{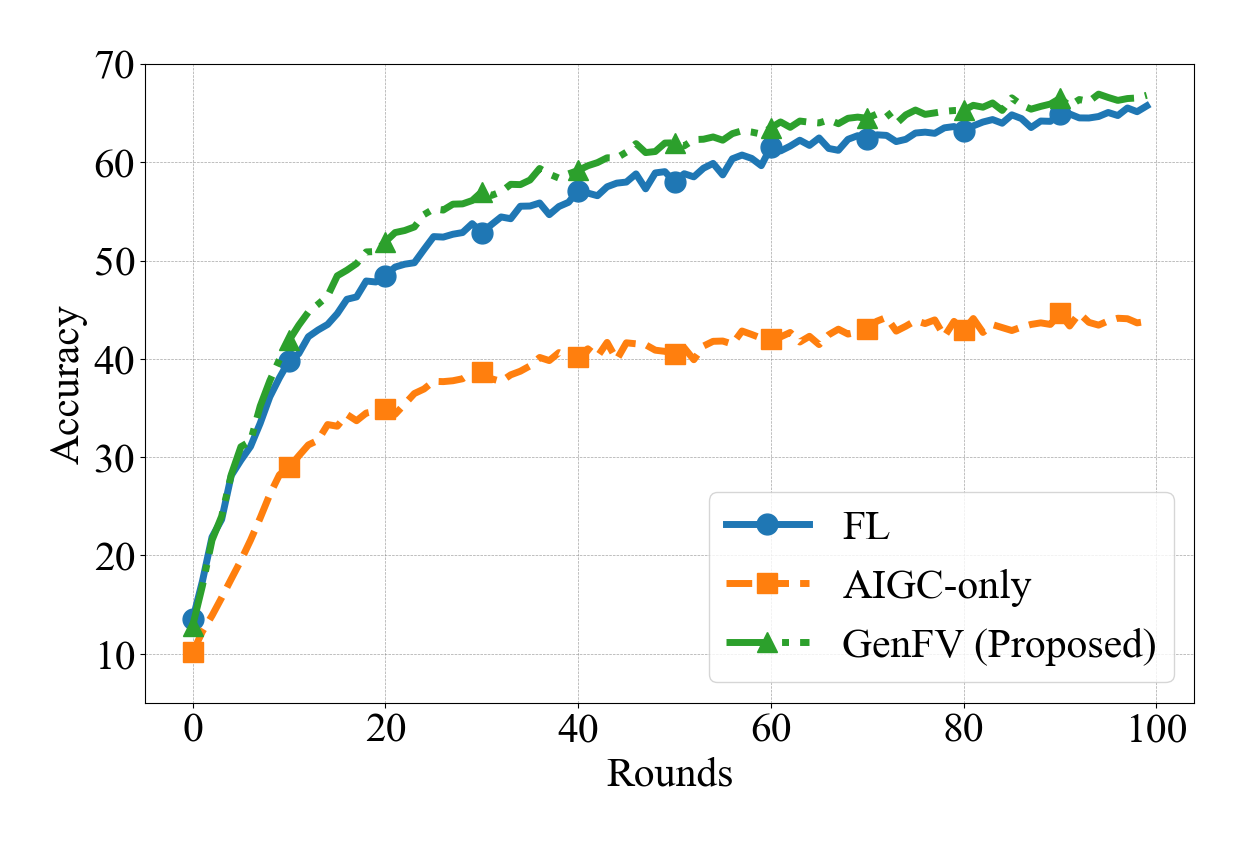}}
    	\hfill
    	\subfloat[Dir($\alpha=1.0$)]{\includegraphics[width=0.23\textwidth]{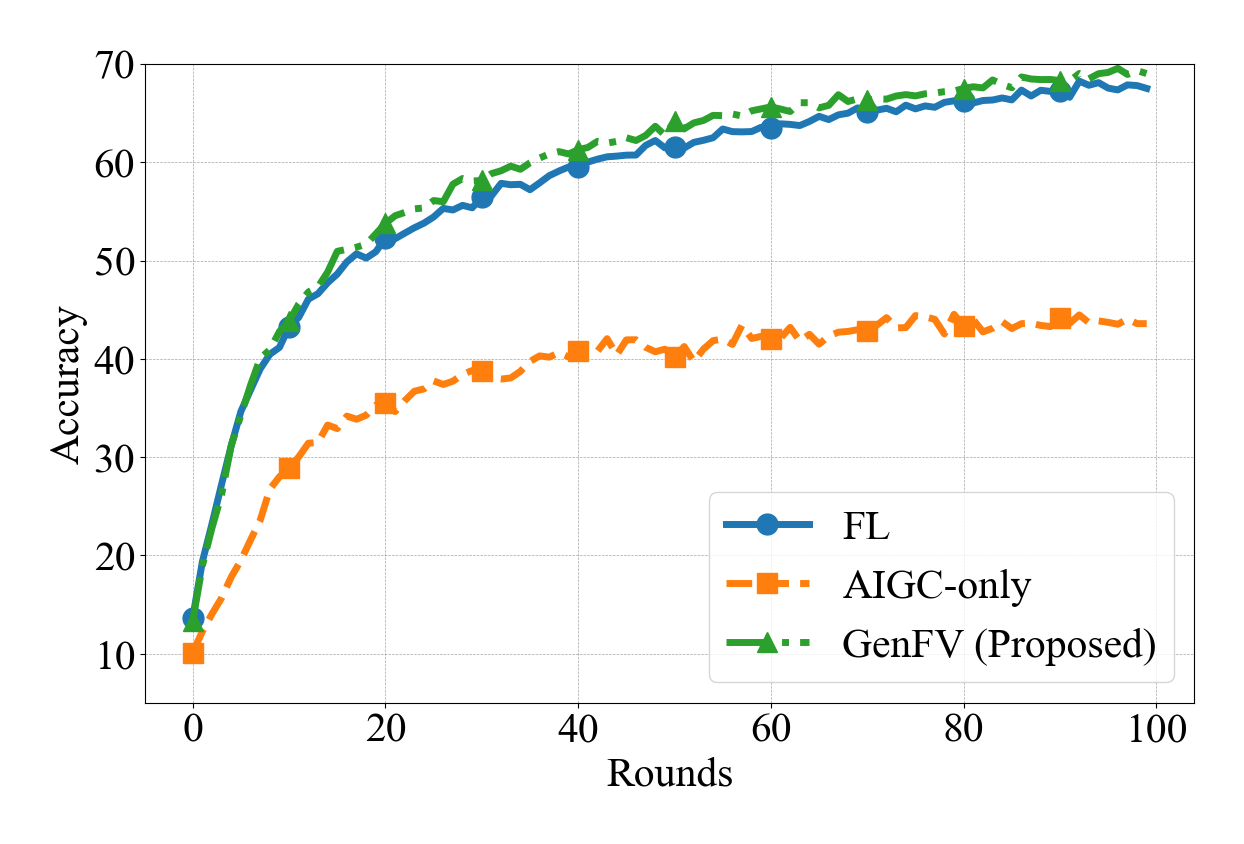}}
    	\caption{Accuracy on CIFAR-10 with different Dirichlet distribution.}
    	\label{fig:cifar10-acc}
    \end{figure*}
    \begin{figure*}[t]
        \centering
        \subfloat[Dir($\alpha=0.1$)]{\includegraphics[width=0.23\textwidth]{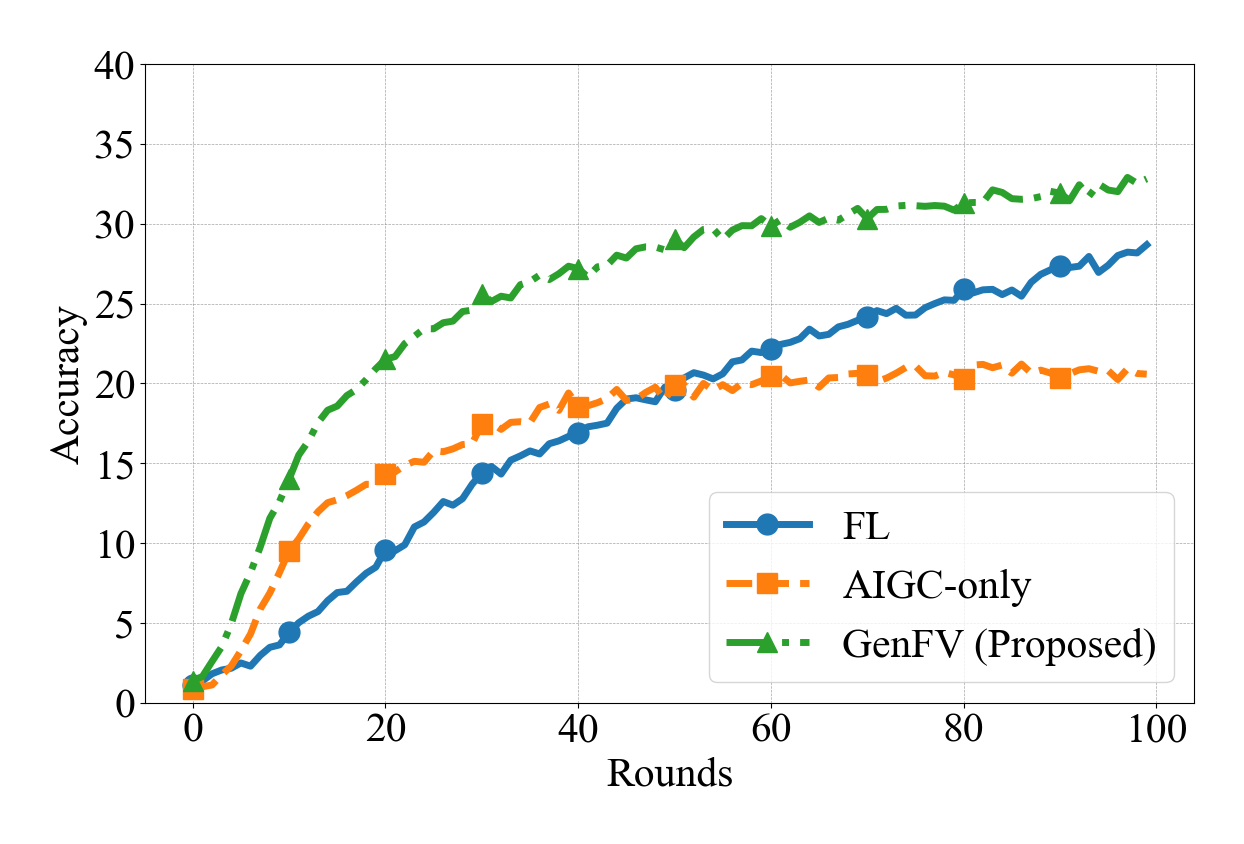}}
        \hfill
        \subfloat[Dir($\alpha=0.3$)]{\includegraphics[width=0.23\textwidth]{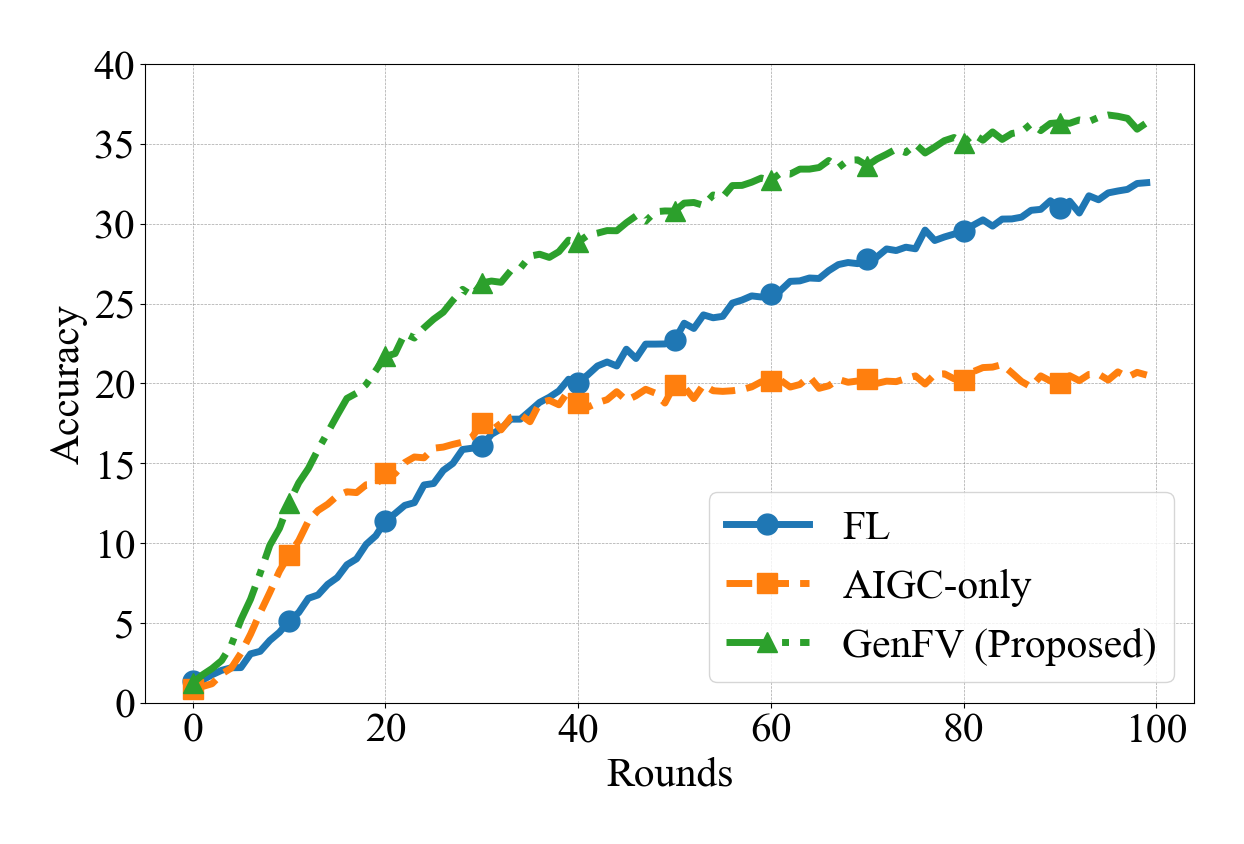}}
        \hfill
        \subfloat[Dir($\alpha=0.5$)]{\includegraphics[width=0.23\textwidth]{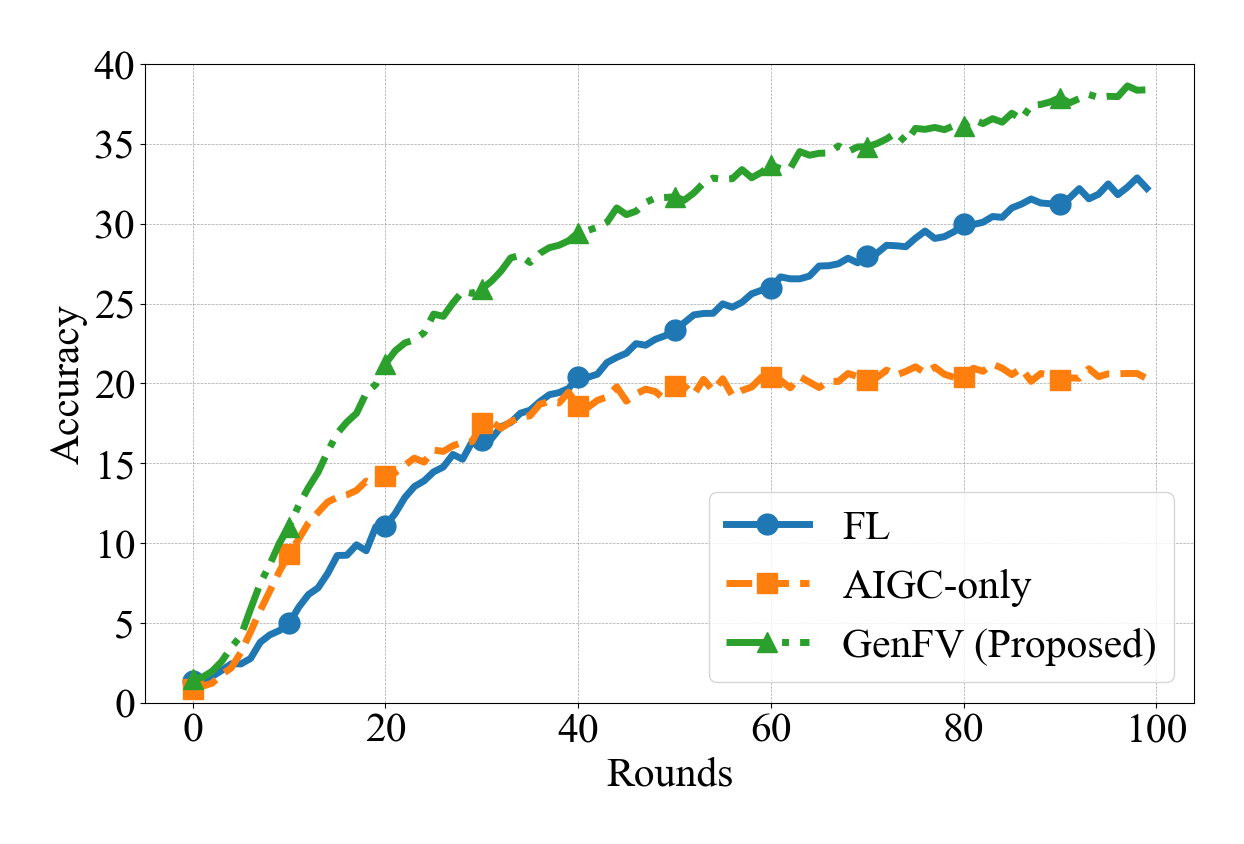}}
        \hfill
        \subfloat[Dir($\alpha=1.0$)]{\includegraphics[width=0.23\textwidth]{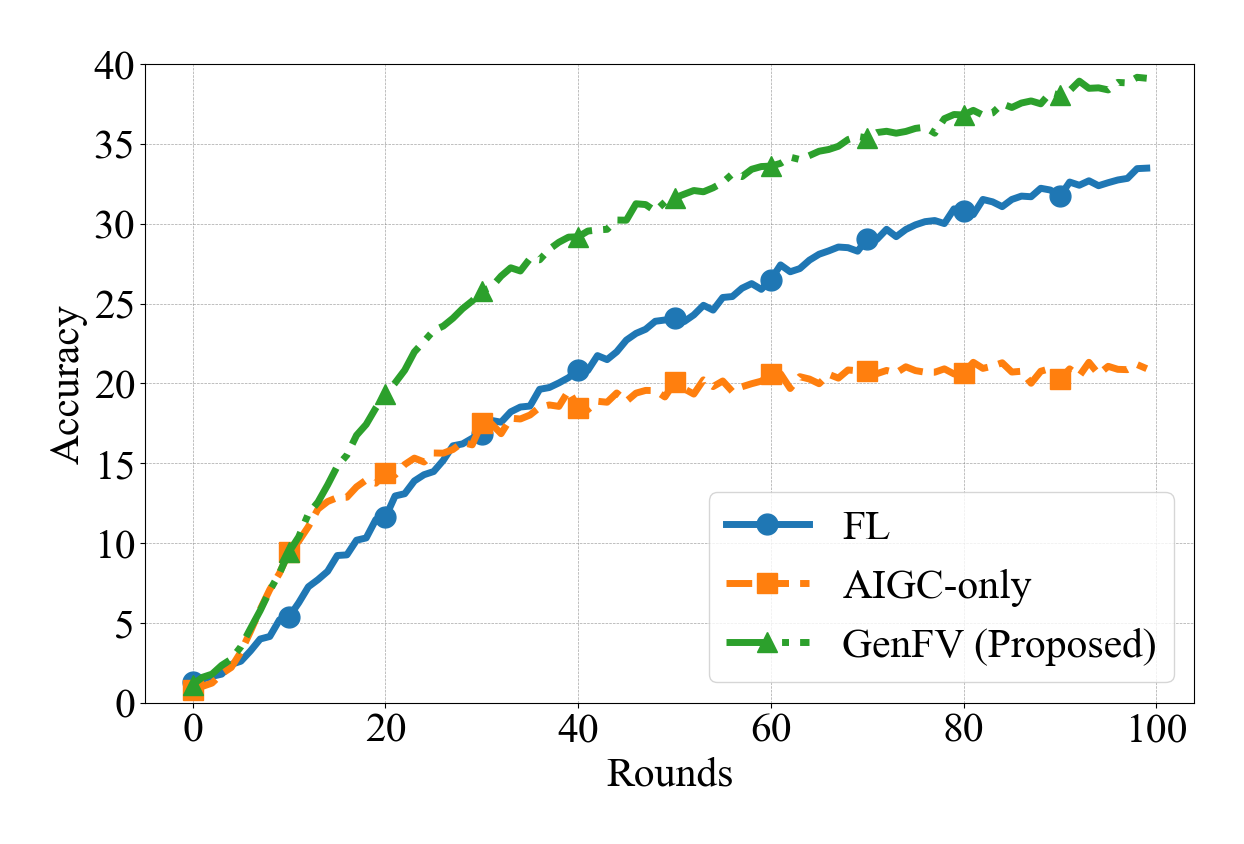}}
        \caption{Accuracy on CIFAR-100 with different Dirichlet distribution.}
        \label{fig:cifar100-acc}
    \end{figure*}	
    \begin{figure*}[t]
        \centering
        \subfloat[Dir($\alpha=0.1$)]{\includegraphics[width=0.23\textwidth]{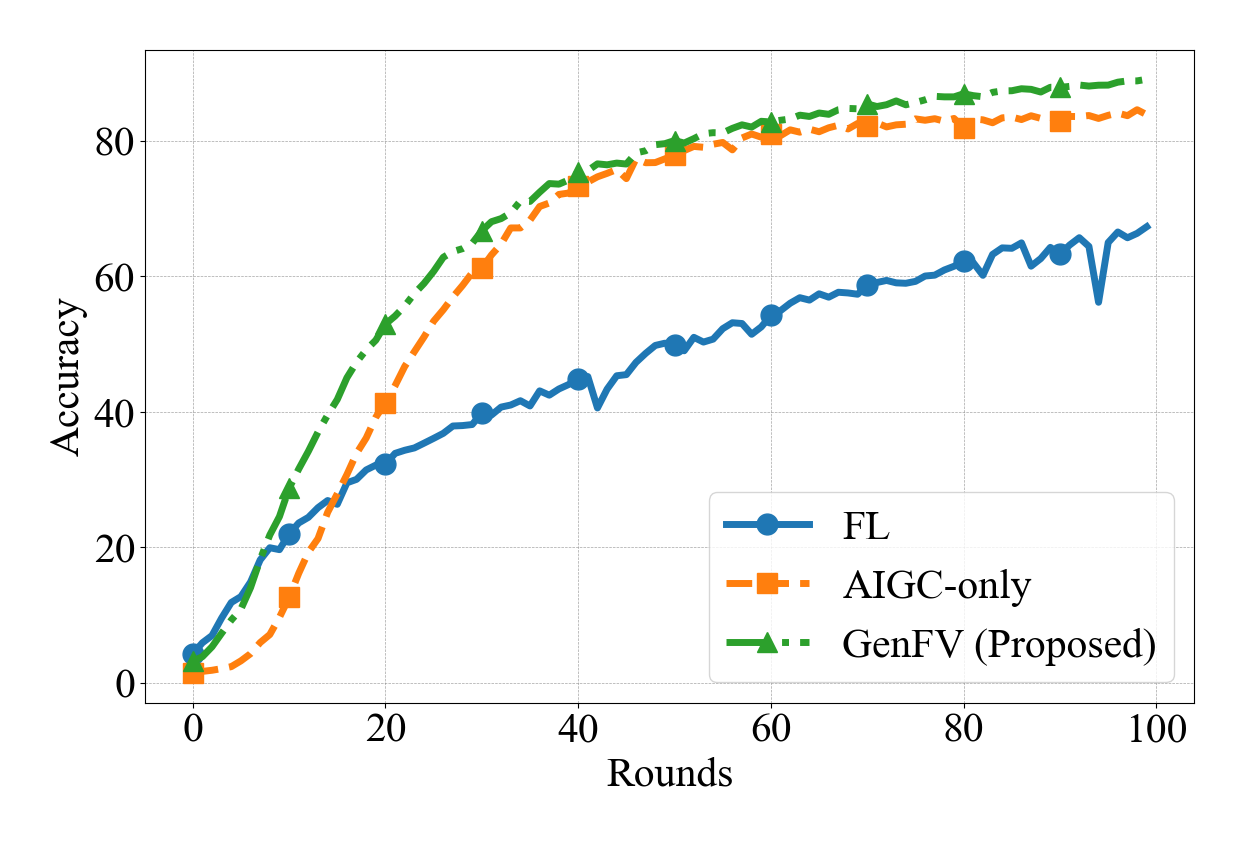}}
        \hfill
        \subfloat[Dir($\alpha=0.3$)]{\includegraphics[width=0.23\textwidth]{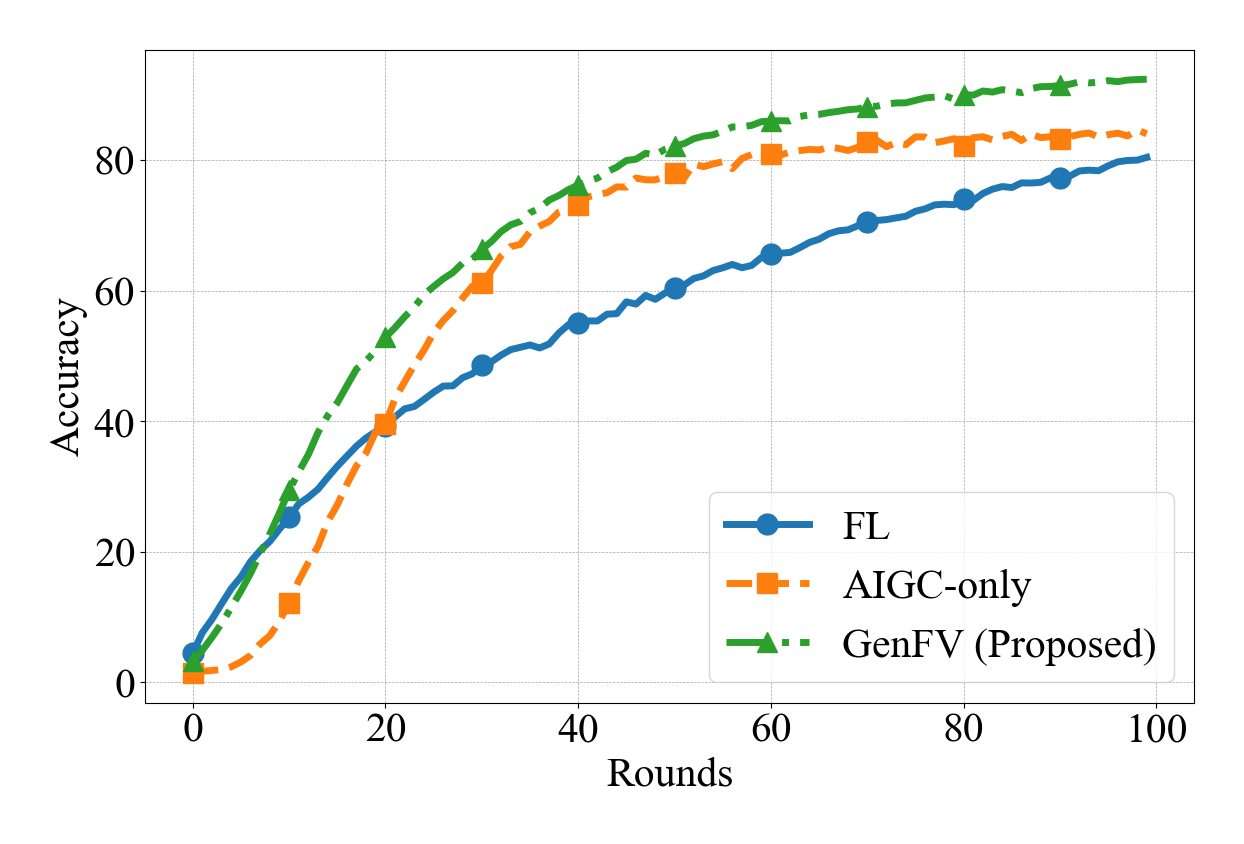}}
        \hfill
        \subfloat[Dir($\alpha=0.5$)]{\includegraphics[width=0.23\textwidth]{images/test_acc_GTSRBds_64bs_0.0001lr_5ep_100r_100u_0.3a_new.png}}
        \hfill
        \subfloat[Dir($\alpha=1.0$)]{\includegraphics[width=0.23\textwidth]{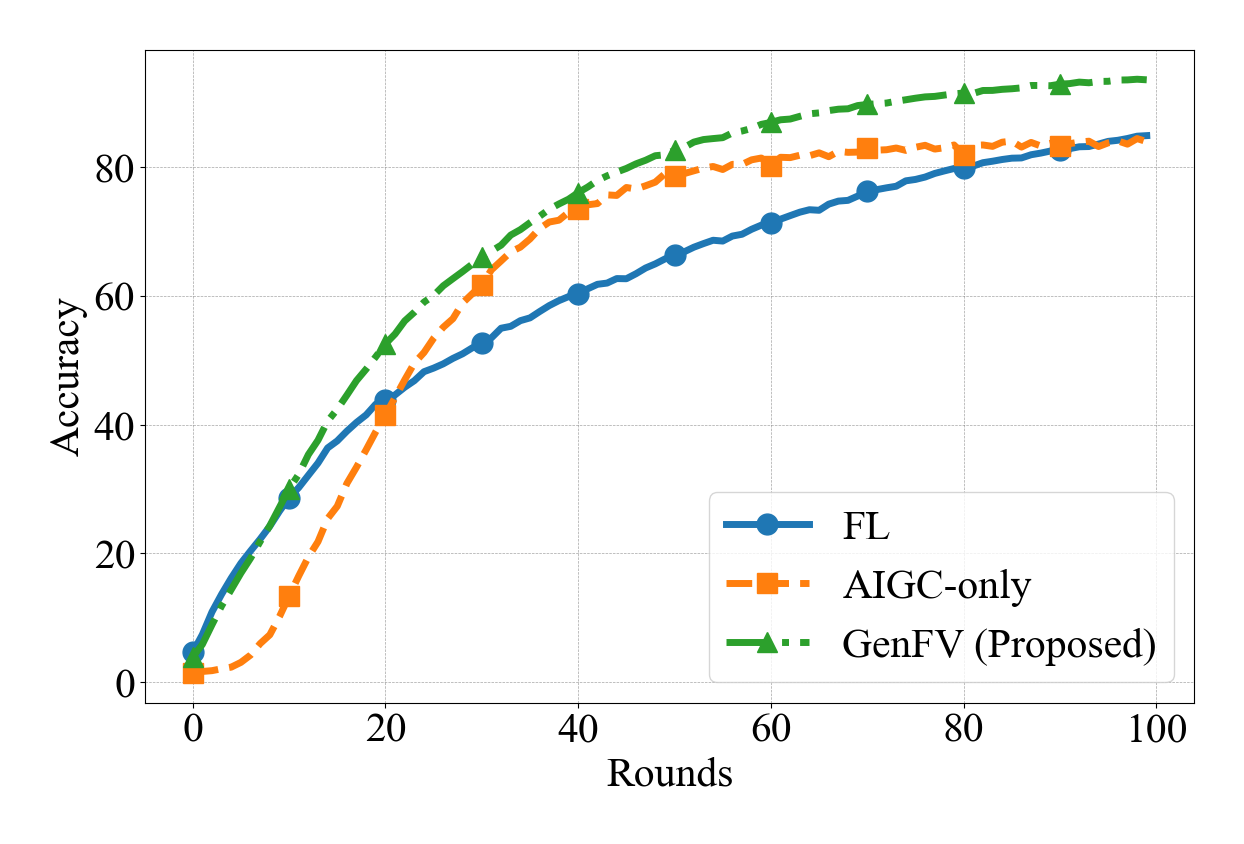}\label{fig:gtsrb-1.0}}
        \caption{Accuracy on GTSRB with different Dirichlet distribution.}
        \label{fig:gtsrb-acc}
    \end{figure*}
\subsection{Resource Allocation Algorithm Evaluation}
Fig. \ref{fig:cifar10-selection-compare} shows the vehicle selection strategy. Four baselines are considered for comparison with GenFV.
\begin{itemize}
    \item FedAvg\cite{mcmahan2017communication}: The server randomly selects vehicles and aggregate local models for global model.
    \item No EMD: The server selects vehicles only with constraint Eqn. (\ref{emd_cons}).
    \item OCEAN-a\cite{9237168}: A long-term client selection and bandwidth allocation problem under finite energy constraints of individual clients.
    \item MADCA-FL\cite{10144680}: A vehicle selection strategy that leverages the success probability (more than 80$\%$) of highly mobile vehicles in completing the training process.
\end{itemize}
Fig. \ref{fig:cifar10-selection-compare} shows FL training loss and testing accuracy with different vehicle selection. These two figures show that all schemes become convergent as the training round increases. We also find that proposed and the other three baselines outperform FedAvg, since the features of vehicles and environments are taken into account in these selection strategies. Compared with No EMD method, show that the necessary of give up some too heterogeneous vehicles. Particularly, we consider the high mobility, dynamics of wireless channel characteristics and data heterogeneity in our proposed method, since vehicles may not complete the training within the allowed latency, which consequently reduces the number of participating vehicles and raise the uncertainty of the FL performance. Besides, it can be observed that proposed selection strategy is superior than No EMD method, because we constraint the heterogeneity of participants' data which can achieve more stable and better results. In addition,  the MADCA-FL method consider the transmission successful probability to avoid the interruption during training is like the No EMD method, neglecting different data distribution of vehicles, caused less performance. OCEAN-a follows an empirical "later-is-better" convergence phenomenon, with increasing selected number of vehicles, it can achieve better performance than Fedavg. These comparisons reveal that proposed vehicle selection strategy achieves the best convergence performance, which demonstrates the effectiveness of our proposed scheme.

In the Fig. \ref{fig:convergencetmax}, each curve represents system performance with different maximum delay constraints ($t_{max}$). As the maximum uplink power increases, the objective value decreases, indicating the reduced delay. Lower $t_{max}$ values result in significantly better performance, achieving lower delays.

Fig. \ref{fig:convergenceselection} shows the impact of after solving different sub problems in the two-scale algorithm when $t_{max}$ is 3.0 seconds. Initially, the objective value is high, but it significantly decreases after solving each sub problem. Particularly after completing $\mathcal{SUBP}$1, $\mathcal{SUBP}$2 and $\mathcal{SUBP}$3, the system delay reaches its lowest. This demonstrates the effectiveness of the two-scale algorithm in optimizing performance through its iterative process.

Fig. \ref{fig:generated_images} illustrates the cumulative number of generated images per label. In each round, the total number of generated images remains consistent under the same wireless conditions. However, CIFAR10 has 10 labels, CIFAR100 has 100 labels, and GTSRB has 43 labels, resulting in a higher number of labels with fewer images per label. As the cumulative total increases, the training time for the augmented model also rises, leading to a decrease in the time allocated for image generation, which ultimately slows the rate of increase.

\subsection{GenFv performance under different non-IID data distribution}

As illustrated in Fig. \ref{fig:cifar10-acc} and Fig. \ref{fig:cifar100-acc}, the proposed GenFV method consistently achieves the highest accuracy across various settings. In contrast, models trained solely on AIGC-generated datasets perform the worst, primarily due to the distributional differences between AIGC-generated and real images. FL based exclusively on local data yields poorer results, mainly due to data heterogeneity. As the parameter $\alpha$ increases and data heterogeneity decreases, the accuracy of FL steadily improves. The results indicate that the proposed method achieves faster convergence at smaller $\alpha$ values. However, this acceleration diminishes as data heterogeneity decreases. We also observe that models trained on data generated by AIGC converge quickly, but still experience bottlenecks in terms of accuracy. This suggests that the capability of AIGC in image generation has significant room for improvement. Furthermore, on the CIFAR-100 dataset, compared to the CIFAR-10 dataset, both the GenFV method and FL show an approximately $10\%$ increase in accuracy at $\alpha=1.0$ compared to $\alpha=0.1$. This clearly reflects the impact of data heterogeneity on federated learning.\\
Fig. \ref{fig:gtsrb-acc} shows the accuracy of the proposed scheme on GTSRB dataset. In these three figures, we can observe that the dataset generated by AIGC-only performs better. Therefore, GenFV benefits significantly from this model augmentation, resulting in remarkable model performance, especially in highly heterogeneous local data scenarios ($\alpha=0.1$). However, as the number of rounds increases, it is predictable that FL will eventually surpass AIGC-only in Fig. \ref{fig:gtsrb-acc}\subref{fig:gtsrb-1.0}. Overall, GenFV combines the advantage of rapid convergence from AIGC-only with the sustained learning capability of FL.

\section{Conclusion}
\label{section7}
	In this paper, we propose a novel framework called AIGC-assisted Federated Learning for Vehicular Edge Intelligence (GenFV). Our approach leverages the diffusion model to address the non-IID data challenge commonly faced in FL. By incorporating the diffusion model into the FL process, we enhance the diversity and representativeness of local datasets across vehicles, even when data distributions are highly skewed or incomplete. We then introduce an augmented FL convergence analysis and propose a weighted policy using the Earth Mover's Distance (EMD) to measure data distribution heterogeneity. Next, we analyze system delay and formulate a mixed-integer nonlinear programming (MINLP) problem to minimize it. To solve this NP-hard problem, we propose a two-scale algorithm. At the large communication scale, we implement labels sharing and select vehicles based on velocity and data heterogeneity. At the small communication scale, we use Successive Convex Approximation and Karush-Kuhn-Tucker conditions to optimally allocate bandwidth and transmission power. Extensive experiments demonstrate that GenFV significantly enhances the performance and robustness of FL in dynamic, resource-constrained environments.\par
   This work demonstrates the potential of integrating the AIGC and FL. Further research is needed to improve GenFV’s security and performance, particularly in strengthening privacy protection during label sharing and enhancing the quality of AIGC-generated images for better real-world adaptation.

\bibliographystyle{IEEEtran}
\bibliography{main}
\end{document}